\documentclass[conference,letterpaper]{IEEEtran}

\usepackage[utf8]{inputenc} 
\usepackage[T1]{fontenc}
\usepackage{url}
\usepackage{ifthen}
\usepackage[cmex10]{amsmath} 

\usepackage{cite}
\usepackage{amssymb,amsfonts}
\usepackage[dvipsnames,table,xcdraw]{xcolor}

\usepackage{graphicx} 
\usepackage{mathtools}
\usepackage{amsthm}
 \usepackage[colorlinks = true, linkcolor = blue, urlcolor=blue, citecolor = red, anchorcolor = green]{hyperref}
\usepackage[capitalize]{cleveref}
\usepackage{multicol}

\newtheorem{definition}{Definition}

\newtheorem{remark}{Remark}

\newtheorem{proposition}{Proposition}
\newtheorem{lemma}{Lemma}

\newcommand{\Tr}[0]{\mathrm{Tr}}
\newcommand{\cA}{\mathcal{A}}

\newcommand{\cC}{\mathcal{C}}
\newcommand{\cD}{\mathcal{D}}

\newcommand{\cH}{\mathcal{H}}
\newcommand{\cI}{\mathcal{I}}

\newcommand{\cL}{\mathcal{L}}
\newcommand{\cM}{\mathcal{M}}
\newcommand{\cN}{\mathcal{N}}

\newcommand{\cP}{\mathcal{P}}
\newcommand{\cQ}{\mathcal{Q}}

\newcommand{\cS}{\mathcal{S}}
\newcommand{\cT}{\mathcal{T}}
\newcommand{\cU}{\mathcal{U}}
\newcommand{\cV}{\mathcal{V}}
\newcommand{\cW}{\mathcal{W}}
\newcommand{\cX}{\mathcal{X}}

\allowdisplaybreaks

\begin{document}

\title{Measured Hockey-Stick Divergence and its Applications to Quantum Pufferfish Privacy} 
\author{%
  \IEEEauthorblockN{Theshani Nuradha\IEEEauthorrefmark{1}, Vishal Singh\IEEEauthorrefmark{2}, and Mark M.~Wilde\IEEEauthorrefmark{1}}
  \IEEEauthorblockA{\IEEEauthorblockA{\IEEEauthorrefmark{1}School of Electrical and Computer
  Engineering, Cornell University, Ithaca, New York 14850, USA.}
  \IEEEauthorrefmark{2} School of Applied and Engineering Physics, Cornell University, Ithaca, New York 14850, USA.}  
}

\maketitle

\begin{abstract} The hockey-stick divergence is a fundamental quantity characterizing several statistical privacy frameworks that ensure privacy for classical and quantum data.
In such quantum privacy frameworks, the adversary is allowed to perform all possible measurements. However, in practice, there are typically limitations to the set of measurements that can be performed. To this end, here, we comprehensively analyze the measured hockey-stick divergence under several classes of practically relevant measurement classes. We prove several of its properties, including data processing and convexity. We show that it is efficiently computable by semi-definite programming for some classes of measurements and can be analytically evaluated for Werner and isotropic states. Notably, we show that the measured hockey-stick divergence characterizes optimal privacy parameters in the quantum pufferfish privacy framework.
With this connection and the developed technical tools, we enable methods to quantify and audit privacy for several practically relevant settings. Lastly, we introduce the measured hockey-stick divergence of channels and explore its applications in ensuring privacy for channels.
\end{abstract}


\section{Introduction}
In classical and quantum information theory, divergences quantify the discrepancy between two probability measures and between two quantum states, respectively. 
Due to the non-commutative nature of operators, a wide spectrum of quantum divergences generalizes their classical counterparts.
 One route for introducing a quantum divergence is  to perform a measurement on the quantum states  and then evaluate the supremum of the  classical divergence of the resulting probability measures over all quantum measurements~\cite{donald1986relative,hiai1991proper,P09}. These are known as measured divergences. 

However, in practice, not all measurements can be implemented due to physical restrictions or design considerations. Hence, it is important to study settings in which only limited measurements are allowed. To this end, Refs.~\cite{matthews2009distinguishability,cheng2023discrimination} analyzed the measured trace distance under restrictive measurements while Ref.~\cite{rippchen2024locally} analyzed locally measured variants of R\'enyi divergences.

In this paper, we focus on measured variants of the hockey-stick divergence under practically relevant classes of measurements. The hockey-stick divergence has recently found application as one of the fundamental quantities characterizing optimal privacy parameters in various classical and quantum privacy frameworks~\cite{asoodeh2021local, hirche2023quantum,angrisani2023unifying,nuradha_QPP,hirche2023quantum_2}. Quantum differential privacy (QDP) is one such privacy framework, which ensures that it is difficult for an adversary to distinguish between two distinct quantum states under \textit{all possible measurements} that an adversary can perform~\cite{QDP_computation17,hirche2023quantum}.

Generalizing QDP, a flexible privacy framework termed quantum pufferfish privacy (QPP) was proposed, where it offers customization and flexibility to incorporate limitations related to the application~\cite{nuradha_QPP}. It generalizes the classical pufferfish privacy framework introduced and analyzed in~\cite{KM14, nuradha2022pufferfishJ,nuradha_MIPP}. One key feature of the QPP framework is the flexibility to specify the \textit{allowed set of measurements}. It was shown in~\cite{nuradha_QPP} that the QPP framework for the special case of \textit{all measurements} is equivalent to a constraint on the hockey-stick divergence,
by generalizing the connection established for QDP in~\cite{hirche2023quantum}.

Inspired by these connections, we study measured variants of the hockey-stick divergence. These variants can further provide insights into the study of privacy in quantum systems with QPP, taking into consideration practically motivated measurement classes. 
In particular, this analysis introduces avenues for quantifying optimal privacy parameters in the QPP setting and auditing privacy by computing the measured hockey-stick divergences, which were not addressed in previous works. 

The main goal of our work is to provide a comprehensive analysis of measured hockey-stick divergences under restricted measurements. 
We first show that, under an additive constraint on the measurement operators (satisfied by all measurement classes that allow classical post-processing), the measured hockey-stick divergence is achieved by two-outcome measurements. 
We then derive key properties satisfied by the measured hockey-stick divergence, including data processing, triangular inequality, and convexity.

Next, we provide computational methods for estimating the measured divergence. We show that the positive-partial-transpose (PPT) measured hockey-stick divergence can be computed by a semi-definite program (SDP). 
We arrive at analytical expressions for Werner and isotropic states under restricted measurements. There, we show that they coincide for different classes of measurements, including PPT measurements and measurements that use local operations and classical post-processing. 

Furthermore, we show that the measured hockey-stick divergence finds operational meaning in $(\varepsilon, \delta)$-QPP as the optimal $\delta$ that can be achieved for a fixed $\varepsilon$ privacy parameter, extending the connection between hockey-stick divergences and QDP established in~\cite{hirche2023quantum}. Then, we utilize the derived results to obtain optimal privacy parameters in specific privacy settings.

Lastly, we introduce the measured hockey-stick divergence for channels. To this end, we provide SDP-computable expressions for the channel divergence under PPT measurements and its connection to the divergence between the Choi states of channels for jointly-covariant channels. With that, we establish analytical expressions for the channel divergence for depolarizing channels. We also highlight an application of these channel divergences in ensuring privacy of channels.

\section{Notations and Definitions} \label{Sec:Notation}
\subsubsection{Notations} A quantum system~$R$ is associated with a finite-dimensional Hilbert space~$\mathcal{H}_R$. $\mathcal{L}(\mathcal{H}_R)$ denotes the set of linear operators acting on $\mathcal{H}_R$. 
We denote the transpose of $X$ by $T(X)$. Let $T_A(C_{AB})$ denote the partial transpose of $C_{AB} \in \mathcal{L}(\mathcal{H}_A \otimes \mathcal{H}_B)$ on the system~$A$. 
We denote the trace of~$C_{AB}$ by $\Tr\!\left[C_{AB} \right]$, and the partial trace of $C$ over the system $A$ by $\Tr_A \!\left[C_{AB}\right]$. $A \geq B$ indicates that $A-B$ is a positive semi-definite (PSD) operator, while $A > B$ indicates that $A-B$ is a positive definite operator, for Hermitian operators $A$ and $B$.

A quantum state $\rho_R\in\mathcal{L}(\mathcal{H}_R)$ 
is a PSD, unit trace operator acting on $\mathcal{H}_R$. We denote the set of all quantum states that act on $\mathcal{H}_R$ as $\mathcal{D}(\mathcal{H}_R)$. 
A quantum channel $\mathcal{N}: \mathcal{L}(\mathcal{H}_A ) \to \mathcal{L}(\mathcal{H}_B)$ is a linear, completely positive, and trace-preserving map (CPTP) from $\mathcal{L}(\mathcal{H}_A)$ to $\mathcal{L}(\mathcal{H}_B)$. We denote the Hilbert--Schmidt adjoint of $\mathcal{N}$ by $\mathcal{N}^\dagger$. A measurement of a quantum system $R$ is described by a
positive operator-valued measure (POVM) $\{M_y\}_{y \in \mathcal{Y}}$, which is defined as a collection of PSD operators satisfying $\sum_{y \in \mathcal{Y}} M_y= I_{R}$, where $I_{R}$ is the identity operator and $\mathcal{Y}$ is a finite alphabet. When applying the above POVM to a state $\rho$, the probability of observing the outcome~$y$ is equal to $\Tr\!\left[M_y \rho \right]$, known as the Born rule.  

\subsubsection{Hockey-Stick Divergences}
Let $\rho$ and $\sigma$ be states.
For $\gamma \geq 0$, the hockey-stick divergence is defined as~\cite{sharma2012strong,nuradha2024contraction}
\begin{align}\label{eq:hockey_stick_all}
    E_\gamma(\rho \Vert \sigma) & \coloneqq \sup_{0 \leq M \leq I} \operatorname{Tr}\!\left[ M(\rho -\gamma \sigma )\right] -(1-\gamma)_+,
\end{align}
where $(x)_+ \coloneqq \max\{0,x\}$. 
In this case, $M$ and $I-M$ are valid measurement operators that satisfy $0 \leq M \leq I $ and $0 \leq I-M \leq I$.
For the classical setting with (discrete) probability distributions $P$ and $Q$, the hockey-stick divergence for $\gamma \geq 0$ is defined as
\begin{align}\label{eq:classical_HS}
    E_\gamma(P \Vert Q) & \coloneqq 
    \sum_x \max\{0, P(x)- \gamma Q(x)\} - (1-\gamma)_+
\end{align}

In this paper, we are interested in restricted measurement operators, which is inspired by the practical feasibility of only certain measurements. Let $\mathcal{M}$ denote the restricted measurement operator set, and let us define
\begin{equation}\label{eq:M_2_set}
    \cM_2 \coloneqq \left\{ M: M, I-M \in \cM \right\}.
\end{equation}
Also, note that $\cM_2=\cM$ in many cases, including all possible measurements and PPT measurements.
With that, we define the measured hockey-stick divergence related to the measurement set $\cM$.

\begin{definition}[Measured Hockey-Stick Divergence] \label{def:HS_Measured}
Let $\rho$ and $\sigma$ be states. For $\gamma \geq 0$, we define the measured hockey-stick divergence based on the measurement set $\cM$ as follows:
\begin{equation}\label{eq:hockey_stick_define}
    E_\gamma^{\mathcal{M}}(\rho \Vert \sigma) \coloneqq \sup_{M\in \mathcal{M}_2} \left\{\operatorname{Tr}\!\left[ M(\rho -\gamma \sigma )\right]\right\} -(1-\gamma)_+,
\end{equation}
where $\cM_2$ is defined in~\eqref{eq:M_2_set}.   
\end{definition}
By choosing $\gamma=1$, we arrive at the measured trace distance, as  used in~\cite{matthews2009distinguishability,cheng2023discrimination}.

An alternative way to define the measured variant of the hockey-stick divergence is as follows: Let $\mu^{\rho,M}$ correspond to the classical probability distribution with atoms $\{ \operatorname{Tr}[M_x \rho]\}_x$,  and let $\mu^{\sigma,M}$ correspond to $\{ \operatorname{Tr}[M_x \sigma]\}_x$ for a POVM $\{M_x\}_x$. 
\begin{multline}\label{eq:hockey_stick_alter}
\widehat{E}_\gamma^{\mathcal{M}}(\rho \Vert \sigma) \coloneqq \\ \sup_{\{M_x \in \mathcal{M}\}_x} \left\{ E_\gamma\!\left( \mu^{\rho,M} \Vert \mu^{\sigma,M} \right):  M_x \geq 0, \ \sum_x M_x = I\right\},
\end{multline}
where the maximization is over $\{M_x \in \mathcal{M}\}_x$ POVMs 
and the classical hockey-stick divergence is defined as in~\eqref{eq:classical_HS}.

\section{Properties of 
Measured Hockey-Stick Divergence} 

\label{Sec:Measured_Divergence}

In this section, we show that the measured hockey-stick divergence defined in~\cref{def:HS_Measured} coincides with the alternative definition in~\eqref{eq:hockey_stick_alter} in many settings, and it satisfies several other properties, including data processing under measurement-compatible channels, the triangular inequality, and convexity.
In this work, we focus on the case when $\gamma \geq 1$, and note that our results extend to $\gamma \geq 0$ by utilizing~\eqref{eq:hockey_stick_define} and applying the equality $E_\gamma^\mathcal{M}(\rho \Vert \sigma) = \gamma E_{1/\gamma}^\mathcal{M}(\sigma \Vert \rho)$ proved in \cref{lem:hockey_stick_gamma_leq}.

Next, we establish conditions under which \cref{def:HS_Measured} coincides with the alternative definition in~\eqref{eq:hockey_stick_alter}.

\begin{proposition} \label{prop:equi_hockey_stick}
Let $\gamma \geq 1$, and let $\rho$ and $\sigma$ be states.
    If $M_1+M_2 \in \mathcal{M}$ holds for all $M_1, M_2 \in \mathcal{M}$ such that $M_1 +M_2 \leq I$  (coarse-graining), then the following equality holds:
\begin{equation}\label{eq:equi_hockey_stick}
        E_\gamma^{\mathcal{M}}(\rho \Vert \sigma) =\widehat{E}_\gamma^{\mathcal{M}}(\rho \Vert \sigma),
    \end{equation}
where $ E_\gamma^{\mathcal{M}}$ is defined in~\cref{def:HS_Measured} and $\widehat{E}_\gamma^{\mathcal{M}}$ is defined in~\eqref{eq:hockey_stick_alter}.
\end{proposition}
\begin{IEEEproof}
See~\cref{App:eqqui_hockey_stick_proof}.
\end{IEEEproof}
\begin{remark}[Measurement Sets Satisfying Equality]
    The coarse-graining process in~\cref{prop:equi_hockey_stick} can be achieved by classical post-processing of the measurement outcome. Therefore, as long as \textit{classical post-processing} of the measurement outcomes is allowed in any set $\mathcal{M}$, then~\eqref{eq:equi_hockey_stick} holds. The set of PPT measurements and local measurements with classical post-processing also satisfy the coarse-graining property (see~\cref{Sec:special_Setting} for details). 
\end{remark}

We define a channel $\cN$ to be $\cM$-compatible under the measurement class $\cM$ if
\begin{equation}
\label{eq:measurement_compatible_channels}
     M\in \cM \implies \cN^\dag(M) \in \cM. 
\end{equation}

\begin{proposition} [Properties of Measured Hockey-Stick Divergence] \label{prop:properties_m_HS}
Let $\cM$ be the set of allowed measurement operators, and suppose that $\rho$ and $\sigma$ are states.
    Then $E_\gamma^\cM$ 
    in~\cref{def:HS_Measured} satisfies the following propertes:
    \begin{enumerate}
        \item Data Processing under $\mathcal{M}$-compatible channels: 
        If $\cN$ satisfies~\eqref{eq:measurement_compatible_channels} under $\cM$,
        then 
        $E_\gamma^\cM(\rho \Vert \sigma) \geq E_\gamma^\cM\!\left( \cN(\rho) \Vert \cN(\sigma) \right)$.
        
        \item Triangular inequality: 
        Let $\gamma_1, \gamma_2 \geq 1$, and suppose that 
        $\tau$ is a state. Then \\
           $ E_{\gamma_1 \gamma_2}^\cM (\rho \Vert \sigma) \leq E_{\gamma_1}^\cM(\rho \Vert \tau) + \gamma_1 E_{\gamma_2}^\cM(\tau \Vert \sigma)$.

        \item Monotonicity in $\gamma$: 
       If $\gamma_1 \geq \gamma_2 \geq 1$, then \\
           $E_{\gamma_1}^\cM (\rho \Vert \sigma) \leq  E_{\gamma_2}^\cM (\rho \Vert \sigma).$

        \item Convexity: Let $\rho \coloneqq \sum_{x \in \cX} p_x \rho_x$ and $\sigma \coloneqq \sum_{x \in \cX} p_x \sigma_x$, where $p, q \in \cP(\cX)$ are probability mass functions and $(\rho_x)_x$ and $(\sigma_x)_x$ are tuples of states. Fix $\gamma \geq 1$. Then, \\
        $E_\gamma^\cM(\rho \Vert \sigma) \leq \sum_{x \in \cX} p_x E_\gamma^\cM( \rho_x \Vert \sigma_x)$.
    \end{enumerate}
\end{proposition}
\begin{IEEEproof} See~\cref{App:properties_proof}.
\end{IEEEproof}
\section{Special Classes of Measurements and States} \label{Sec:special_Setting}
In this section, we consider special classes of measurements, $\cM$, and states in~\cref{def:HS_Measured}, such that $E_\gamma^\cM$ is computable analytically or by using a semi-definite program (SDP).

Let us define several special classes of measurements. 

The set of all measurements is defined as 
$\cM_{\operatorname{ALL}} \coloneqq \left\{ M: 0 \leq M \leq I\right\}$.
The set of local measurements (LO) is defined as 
 $\cM_{\operatorname{LO}} \coloneqq \{ E_A^{x} \otimes F_B^{y}: 
    \{E_A^x\}_x \textnormal{ and }  \{F_B^y\}_y \textnormal{ each form a POVM} \}$.
The set of local operations with classical post-processing is defined as 
\begin{equation}
    \cM_{\operatorname{LO}^\star} \coloneqq \left\{\begin{array}{c}
  \sum_{x,y}  T(x,y) E_A^{x} \otimes F_B^{y}:   \\
 \{E_A^x\}_x \textnormal{ and } \{F_B^{y}\}_y \textnormal{ each form a POVM}, \\  0 \leq T(x,y) \leq 1 \ \forall x,y
    \end{array}
\right\}.
\end{equation}
The set of local operations and one-way classical communication (1W-LOCC) based measurements is defined as 
\begin{equation}
    \cM_{\operatorname{1W-LOCC}} \coloneqq \left\{\begin{array}{c}
  \sum_x E_A^{x} \otimes F_B^{x,y}:   
 \{E_A^x\}_x \textnormal{ and } \\ \{F_B^{x,y}\}_y  \ \forall x \ \textnormal{each form a POVM} 
    \end{array}
\right\}.
\end{equation}
This set can also be understood as consisting of measurement operators that correspond to applying a 1W-LOCC channel and then  a local measurement followed by classical post-processing. 
Similarly, the local operations and classical communications (LOCC) measurement set $\cM_{\operatorname{LOCC}}$ is comprised of measurement operators generated by applying an LOCC channel followed by a local measurement and classical post-processing.
The set of PPT measurements is defined as
\begin{equation}
   \mathcal{M}_{\operatorname{PPT}}= \{M_{AB}: 0\leq M_{AB}, \mathrm{T}_B(M_{AB}) \leq I \} .
\end{equation}
By choosing $\cM=\mathcal{M}_{\operatorname{PPT}}$, we also have that $\mathcal{M}_{\operatorname{PPT}}=\cM_2$, where $\cM_2$ is defined in~\eqref{eq:M_2_set}. This follows because $M \in \cM_{\operatorname{PPT}}$ implies that $I-M \in \cM_{\operatorname{PPT}}$.

\cref{prop:equi_hockey_stick} implies that $ E_\gamma^{\mathcal{M}}(\rho \Vert \sigma) =\widehat{E}_\gamma^{\mathcal{M}}(\rho \Vert \sigma)$ for $\cM \in \{ \cM_{\operatorname{LO}^\star}, \cM_{\operatorname{1W-LOCC}}, \cM_{\operatorname{LOCC}}, \cM_{\operatorname{PPT}}\}$.
For the choice of PPT measurements, the measured hockey-stick divergence can be formulated as an SDP as follows.

\begin{proposition}[PPT measured Hockey-Stick Divergence] \label{prop:PPT_Measured_HS_SDP}
Let $\gamma \geq 1$, and let $\rho, \sigma \in \cD\!\left( \mathcal{H}_A \otimes \mathcal{H}_B\right)$. 
   The quantity $E_\gamma^{\operatorname{PPT}}$ can be expressed as the following SDPs:
    \begin{multline}
\sup_{0 \leq M_{AB} \leq I}\left\{  \operatorname{Tr}\!\left[ M_{AB}(\rho -\gamma \sigma )\right]: 0\leq \mathrm{T}_B(M_{AB}) \leq I   \right\} =\\
      \inf_{\substack{Y_i \geq 0}} 
        \left\{ \operatorname{Tr}[Y_4 +Y_3] : Y_3 - Y_1 +T_B(Y_4 -Y_2) \geq \rho -\gamma \sigma\right\}. 
    \end{multline}
\end{proposition}
\begin{IEEEproof}
See Appendix~\ref{app:PPT_Measured_HS_SDP}.
\end{IEEEproof}

\medskip
The above SDPs can be used to provide upper bounds on the locally measured hockey-stick divergence when the measurement set is comprised of LO, ${\operatorname{LO}^\star}$, 1W-LOCC, and LOCC, because
$E_\gamma^{\operatorname{LO}}(\rho \Vert \sigma) \leq   E_\gamma^{{\operatorname{LO}^\star}}(\rho \Vert \sigma) 
 \leq E_\gamma^{\operatorname{LOCC}}(\rho \Vert \sigma) \leq    E_\gamma^{\operatorname{PPT}}(\rho \Vert \sigma)$
 for all states $\rho$ and $\sigma$.

Next, we consider Werner states \cite{Wer89}, for which we characterize the measured hockey-stick divergence analytically.
We define a Werner state as follows:
\begin{equation}
    \omega^p \coloneqq p \Theta + (1-p) \Theta^\perp,
\end{equation}
where $p\in [0,1]$, and 
\begin{align} \label{eq:werner_extremes}
    \Theta  \coloneqq \frac{I+F}{d(d+1)}, \quad 
    \Theta^\perp  \coloneqq \frac{I-F}{d(d-1)}.
\end{align}
with $F \coloneqq \sum_{i,j=1}^d |i \rangle\!\langle j| \otimes  |j \rangle\!\langle i| $.

\begin{proposition}[Hockey-Stick Divergence for Werner States]\label{prop:HS_Werner_general}
    Let $p,q \in [0,1]$. For $\gamma \geq 1$, we have that 
\begin{equation}
    E_\gamma(\omega^q \Vert \omega^p) 
    =\max\{0,q-\gamma p, (1-q) -\gamma (1-p)\}.
\end{equation} 
\end{proposition}
\begin{IEEEproof}
See~\cref{App:Werner_States}.
\end{IEEEproof}
\begin{proposition}[Measured Hockey-Stick Divergence for Werner States]
\label{prop:Measured_HS_Werner}
Let $p,q \in [0,1]$.
     The following holds for $\gamma \geq 1$ and $\cM \in \left\{ \cM_{\operatorname{LO}^\star}, \cM_{\operatorname{1W-LOCC}}, \cM_{\operatorname{LOCC}}, \cM_{\operatorname{PPT}}\right\}$:
    \begin{equation}
 E_\gamma^{\cM}(\omega^q\Vert \omega^p)    =  \max\left\{0, \frac{2(q-\gamma p)} {d+1}, 1-\gamma- \frac{2(q-\gamma p)} {(d+1)}\right\}.
    \end{equation}
\end{proposition}
\begin{IEEEproof}
We prove the inequality ``$\geq$'' by specific choices of measurement operators that are $\operatorname{LO}^\star$: $M=0$, $M= \sum_{i=1}^d |i\rangle\!\langle i| \otimes |i\rangle\!\langle i|$, and $M=\sum_{i\neq j=1}^d |i\rangle\!\langle j| \otimes |i\rangle\!\langle j| .$ 
Then to obtain the inequality ``$\leq$'', we utilize the symmetry of Werner states and the PPT measurement operators. 
See~\cref{App:Werner_States} for a complete proof.
\end{IEEEproof}

\begin{remark}[High Dimensions]
   When $d \to \infty$ in \cref{prop:Measured_HS_Werner}, we see that the measured hockey-stick divergence converges to zero, while with all measurements, it can still be strictly positive, as shown in~\cref{prop:HS_Werner_general}. This highlights that Werner states are distinguishable with all measurements, while with limited measurements, that distinction  decays with increasing $d$. We utilize this property in~\cref{prop:QPP_Werner}.
\end{remark}

\begin{remark}[Isotropic States] \label{Rem:Isotropic_E_g}
    We obtain an analytical expression for the measured hockey-stick divergence of isotropic states in~\cref{App:Isotrpic_States}. They also coincide for different classes of measurements, similar to Werner states. 
\end{remark}

\section{Applications to Privacy}

In this section, we show how the measured hockey-stick divergence finds use in ensuring privacy for states, where  privacy is imposed by quantum pufferfish privacy introduced in~\cite{nuradha_QPP}. 
We also utilize the tools developed in Sections~\ref{Sec:Measured_Divergence} and~\ref{Sec:special_Setting} to study the privacy of quantum systems.

For simplicity, let us consider a special case of quantum pufferfish privacy (QPP). We call this special setting restricted quantum local differential privacy (QLDP).  
Note that the  results below equally apply to the general QPP setting with the appropriate optimizations \cite[Definition~4]{nuradha_QPP}. 

\begin{definition}[$(\cM,\cS)$-Restricted QLDP for States] \label{def:restricted_QLDP}
    Fix $\varepsilon \geq 0$ and $\delta \in [0,1]$. Let $\cS$ be a subset of quantum states (i.e., $\cS \subseteq \mathcal{D}(\mathcal{H})$).
    Let $\mathcal{A}$ be a quantum algorithm (viz., a quantum channel). Then~$\mathcal{A}$ is $ (\varepsilon, \delta)$-local differentially private under $(\mathcal{M},\cS)$ if  the following condition is satisfied:
\begin{equation} \Tr\!\left[M \mathcal{A}(\rho)\right] \leq e^\varepsilon \Tr\!\left[M \mathcal{A}(\sigma)\right] + \delta ,
\label{eq:QLDP-def-M}
\end{equation}
for all  $\rho, \sigma \in \mathcal{S}$ and all $ M \in  \mathcal{M}_2 $,
where $\cM_2$ is defined in~\eqref{eq:M_2_set} with the chosen $\cM$.
\end{definition}
\begin{proposition}[Equivalent Representation]\label{prop:QLDP_Equi_M}
 $ (\varepsilon, \delta)$-local differential privacy under $(\mathcal{M}, \cS)$ of a mechanism $\mathcal{A}$  (as defined in~\cref{def:restricted_QLDP}) is equivalent to 
    \begin{equation}
        \sup_{\rho, \sigma \in \cS} E_{e^\varepsilon}^{\mathcal{M}}\!\left( \mathcal{A}(\rho) \Vert \mathcal{A}(\sigma) \right) \leq \delta.
    \end{equation}
\end{proposition}
\begin{IEEEproof}
   This follows by rearranging terms in~\cref{def:restricted_QLDP} and recalling \cref{def:HS_Measured} with $\gamma =e^\varepsilon \geq 1$. 
\end{IEEEproof}
\begin{remark}[Operational Interpretation]
    \cref{prop:QLDP_Equi_M} provides an operational interpretation for the measured hockey-stick divergence $E_{e^\varepsilon}^\cM$ as the optimal $\delta$ for fixed $\varepsilon$ when the adversary is allowed to use measurements belonging to the set $\cM$. 
\end{remark}

\begin{proposition}[QLDP for Werner States]\label{prop:QPP_Werner}
    Let $\cS= \{ \omega^p: 0\leq p \leq 1\}$. Then, for $\cM={\cM_{\operatorname{ALL}}}$, the identity channel is not private; i.e., $\delta=1$ for all $\varepsilon \geq 0$. 

    Furthermore, the identity channel is 
    $(\varepsilon, 2/(d+1))$- QLDP under $(\cM,\cS)$ for $\varepsilon \geq 0$ and for $\cM \in \{\cM_{\operatorname{LO}^\star}, \cM_{\operatorname{1W- LOCC}}, \cM_{\operatorname{LOCC}}, \cM_{\operatorname{PPT}}\}$.
\end{proposition}

\begin{IEEEproof}
    The proof follows by choosing $\cA= \cI$ in~\cref{prop:QLDP_Equi_M}, using the convexity of hockey-stick divergence in~\cref{prop:properties_m_HS}, and applying Propositions~\ref{prop:HS_Werner_general} and~\ref{prop:Measured_HS_Werner}.
\end{IEEEproof}

\medskip

This shows explicitly how the measurement class plays an inherent role in ensuring privacy, if we are aware of the measurements that the adversary can perform. In particular, an algorithm will be completely non-private with a certain class of measurements (all measurements in~\cref{prop:QPP_Werner}) and almost completely private for another class of measurements (LOCC measurements in~\cref{prop:QPP_Werner} with increasing dimension).

\begin{remark}[Privacy Auditing]
    Auditing privacy refers to identifying whether an algorithm is private, as demanded by~\cref{def:restricted_QLDP}. In many practical scenarios, the set $\mathcal{S}$ relevant to an $(\mathcal{M},\mathcal{S})$-restricted QLDP setting contains a finite number of elements, which allows for a computationally feasible method to audit privacy. For a channel $\mathcal{A}$, a measurement set $\mathcal{M} \subseteq \mathcal{M}_{\operatorname{PPT}}$, and a set of states $\mathcal{S}$ with finite number of elements, one can compute $E^{\operatorname{PPT}}_{e^{\varepsilon}}\!\left(\mathcal{A}\!\left(\rho\right),\mathcal{A}\!\left(\sigma\right)\right)$ for each pair $\left(\rho,\sigma\right) \in \mathcal{S}\times \mathcal{S}$ via an SDP. If $\max_{\rho,\sigma \in \mathcal{S}}E^{\operatorname{PPT}}_{e^{\varepsilon}}\!\left(\mathcal{A}\!\left(\rho\right),\mathcal{A}\!\left(\sigma\right)\right) \le \delta$, then the mechanism $\mathcal{A}$ is guaranteed to be $(\varepsilon,\delta)$-restricted QLDP under $(\mathcal{M},\mathcal{S})$.

    Although the method described above provides a computationally feasible way to verify if a mechanism is private, the time complexity of this algorithm is exponential in the number of qubits and polynomial in the cardinality of $\mathcal{S}$. This calls for more efficient algorithms to estimate the measured hockey-stick divergence and its optimization over a given set of states, which we leave for future investigation.
\end{remark}

It was shown in~\cite{nuradha2024contraction,Christoph2024sample} that the contraction of quantum divergences under the processing of private channels is an essential tool in the study of statistical tasks under privacy constraints. Under restricted QLDP constraints, we obtain the following contraction of the measured trace distance.
\begin{proposition}[Contraction of Trace Distance under Restricted QLDP] \label{Cor:contraction_trace_distance}
    Let $\cA \in \mathcal{B}^{\varepsilon,\delta}_{\cM,\cS}$
    suppose that $\cA$ is $(\varepsilon,\delta)$-QLDP under $(\cM, \cS)$ defined in~\cref{def:restricted_QLDP}.
    Then for all states $\rho$ and $\sigma \in \cS$, we have that
    \begin{align}
       E_{1}^\cM\!\left( \cA(\rho) \Vert \cA(\sigma) \right)  \leq \frac{e^\varepsilon -1 + 2\delta}{e^\varepsilon +1}.
    \end{align}
\end{proposition}

\begin{IEEEproof}
   The proof follows because $  \left(\gamma +1\right)   E_1^\cM(\rho \Vert \sigma) \leq E_\gamma^\cM(\rho \Vert \sigma) + E_\gamma^\cM(\sigma \Vert \rho) + \gamma -1$ (shown in~\cref{lem:E_1_gamma}),
 substituting $\gamma= e^\varepsilon$, and applying~\cref{prop:QLDP_Equi_M}.
\end{IEEEproof}

\section{Measured Channel Divergences}
In this section, we extend the study of hockey-stick divergence to quantum channels.

\begin{definition}[Hockey-Stick Divergence for Channels ]
  Let $\cP_{A\to B}$, $\cQ_{A \to B}$ be channels, and let $\cM$ be the allowed set of measurements on systems~$R$ and $B$. Then for $\gamma \geq 0$, the channel divergence for the measured hockey-stick divergence is defined as    \begin{equation}\label{eq:def_channel_HS}
    E^\cM_\gamma\!\left(\cP \Vert \cQ \right) \coloneqq \sup_{\rho_{RA}} E_\gamma^\cM\!\left( \cP(\rho_{RA}) \Vert \cQ({\rho_{RA})} \right),
\end{equation}
where the optimization is over all states with unbounded reference system $R$.
\end{definition}
Using the Schmidt decomposition theorem and the convexity of measured hockey-stick divergence in~\cref{prop:properties_m_HS}, the optimization in~\eqref{eq:def_channel_HS} can be restricted to pure states with the $R$ system isomorphic to the $A$ system.

Next, we show how the channel divergence of the measured hockey-stick divergence appears in  applications related to ensuring privacy for quantum channels and channel discrimination of two quantum channels.

\textbf{Privacy for Channels:}
Previously, \textit{privacy of states} has been considered under the QDP and QPP frameworks~\cite{QDP_computation17,hirche2023quantum,nuradha_QPP}. However, there are also scenarios where one needs to privatize the channel used in the quantum operation. These will find use in quantum reading \cite{Pirandola2011}, where encoding different values correspond to applying different channels. In this setting, we require the channels to remain indistinguishable to the adversary. To accomplish that, we define a privacy framework that ensures \textit{privacy for channels}. This is accomplished via superchannels, which are linear maps that transform one channel to another channel~\cite{chiribella2008transforming,gour2019comparison}.

\begin{definition}[$(\cM, \cC)$-QLDP for Channels]\label{def:QLDP_channels}
    Fix $\varepsilon \geq 0$ and $\delta \in [0,1]$. 
    Let $\Theta$ be a superchannel,  and let $\cC$ be a subset of channels from $\cL(\cH_A) \to \cL(\cH_B)$. The algorithm/supermechanism~$\Theta$ is $ (\varepsilon, \delta)$-local differentially private under $(\mathcal{M},\mathcal{C})$ if for all $\cP_{A \to B}$,$\cQ_{A \to B} \in \cC$  the following condition is satisfied for all  $\rho_{RA}  \in \mathcal{D}(\mathcal{H}_R \otimes \cH_A)$ and for all $ M \in  \mathcal{M}_2 $:
\begin{equation} \Tr\!\left[M \left(\Theta(\cP)(\rho_{RA} )\right)\right] \leq e^\varepsilon \Tr\!\left[M\left( \Theta(\cQ)(\rho_{RA} )\right)\right] + \delta ,
\label{eq:QLDP-def-M_Channel}
\end{equation}
where $\cM_2$ is defined in~\eqref{eq:M_2_set} with the chosen $\cM$.
\end{definition}
Similar to \cref{prop:QLDP_Equi_M}, we next show that there is an equivalent representation for~\cref{def:QLDP_channels} using the channel hockey-stick divergence. 
\begin{proposition}
    [Equivalent Representation]\label{prop:QLDP_Equi_M_Channels}
 $ (\varepsilon, \delta)$-local differential privacy under $(\mathcal{M}, \cC)$ of a supermechanism $\Theta$  (as defined in~\cref{def:QLDP_channels}) is equivalent to 
    \begin{equation}
        \sup_{\cP, \cQ \in \cC} E_{e^\varepsilon}^{\mathcal{M}}\!\left( \Theta (\cP) \Vert \Theta(\cQ) \right) \leq \delta.
    \end{equation}
\end{proposition}

\begin{IEEEproof}
    This follows by rearranging terms in~\cref{def:QLDP_channels},  recalling \eqref{eq:def_channel_HS}, and applying  \cref{def:HS_Measured} for $\gamma =e^\varepsilon \geq 1$.
\end{IEEEproof}

\medskip 

Notably, \cref{prop:QLDP_Equi_M_Channels} shows that the channel hockey-stick divergence is the fundamental quantity in the task of guaranteeing privacy of quantum channels, thus providing an operational interpretation for it.

\subsubsection{Special Measurements and Channels}
Here, we focus on the channel divergence by considering special classes of measurements, including PPT measurements and special classes of channels, namely, jointly-covariant channels.

    For the class of all measurements, the channel divergence is SDP computable. For depolarizing channels, we derive analytical expressions for it. See \cref{App:channel_div_all_meas} for a detailed study. 
We next show that the channel divergence under PPT measurements is also SDP computable. 
\begin{proposition}[Channel Divergence under PPT Measurements] \label{prop:channel_div_PPT}
 Let $\cP_{A \to B}$ and $\cQ_{A \to B}$ be two quantum channels, and let $\gamma \geq 1$. Then 
$E_\gamma^{\operatorname{PPT}}(\cP \Vert \cQ)$  is equal to
    \begin{align}
    &\sup_{ \rho_{R},\Omega_{RB} \geq 0} \left\{ \begin{array}
[c]{c}\Tr\!\left[ \Omega_{RB} (\Gamma_{RB}^\cP - \gamma \Gamma_{RB}^\cQ )\right] :\\\begin{aligned}
           & \Tr[ \rho_R]=1, \ \Omega_{RB} \leq \rho_R \otimes I_B \\ 
           & T_B[\Omega_{RB}] \geq 0, \ T_B[\Omega_{RB}] \leq \rho_R \otimes I_B
        \end{aligned}
        \end{array}
        \right\} =\\
        &\inf_{ \substack{\mu \geq 0, Z_{RB} \geq 0 \\ L_{RB}, Y_{RB} \geq 0} }\left\{ \mu : \begin{aligned}
            &\Gamma_{RB}^\cP - \gamma \Gamma_{RB}^\cQ \leq Z_{RB} +T_B(Y_{RB}-L_{RB}) \\ & \mu I_{RB} \geq \Tr_B[ Z_{RB}] + \Tr_B[Y_{RB}] 
         \end{aligned}\right\},
    \end{align}
    where $\Gamma_{RB}^\cN$ is the Choi operator of the channel $\cN_{A \to B}$. 
\end{proposition}
\begin{IEEEproof}
See~\cref{App:channel_div_PPT_proof}.
\end{IEEEproof}
\medskip


Let $G$ be a finite group, and for every $g \in G$, let $g \to U_A(g)$ and $g \to V_B(G)$ be unitary representations acting on the input and output spaces of the channel, respectively. A quantum channel $\cN_{A \to B}$ is covariant with respect to these representations if the following equality holds for every input density operator $\rho_A \in \cL(\cH_A)$ and group element $g \in G$: 
    $\cN_{A \to B}\!\left(U_A(g) \rho_A U_A^\dag(g)\right) =V_B(g) \cN_{A \to B}(\rho_A)V_B^\dag(g)$.

For covariant quantum channels, the input states  can be restricted  to a specific form in the optimization in~\eqref{eq:def_channel_HS} and can further be characterized by the Choi states of the channels if the representation $\{ U_A(g)\}_{g \in G}$ is irreducible.
\begin{proposition} \label{prop:PPT_measured_joint_covariance}
    Let $G$ be a finite group with $\{ U_A(g)\}_{g\in G}$ and $\{V_B(g)\}_{g \in G}$ unitary representations. Let $\cP_{A \to B}$ and $\cQ_{A \to B}$ be quantum channels that are covariant with respect to the group $G$. Then $E_\gamma^{\operatorname{PPT}}(\cP \Vert \cQ)$ is equal to
    \begin{align} \label{eq:reduced_form_chan_PPT}
\sup_{\substack{\psi_{RA} :\\ \psi_A = \cT_G (\psi_A)}}\left\{ E_\gamma^{\operatorname{PPT}}\!\left(\cP_{A \to B}(\psi_{RA}) \Vert \cQ_{A \to B} (\psi_{RA})\right) 
          \right\},
    \end{align}
where system $R$ is isomorphic to system $A$, $\psi_A \coloneqq \Tr_{R}[\psi_{RA}]$, and 
   $ \cT_G (\omega_A) \coloneqq \frac{1} {|G|} \sum_{g \in G} U_A(g) \omega_A U_A^\dag(g).$

Furthermore, if the representation $\{ U_A(g)\}_{g \in G}$ is irreducible, then an optimal state in~\eqref{eq:reduced_form_chan_PPT} is the maximally entangled state $\Phi_{RA}$, so that
$E_\gamma^{\operatorname{PPT}}(\cP \Vert \cQ)=  E_\gamma^{\operatorname{PPT}}\!\left( \cP_{A \to B}( \Phi_{RA}) \Vert \cQ_{A \to B}( \Phi_{RA}) \right)$.
\end{proposition}
\begin{IEEEproof}
See~\cref{App:PPT_measured_joint_cov_Proof}. The proof relies on the data-processing of $E_\gamma^{\operatorname{PPT}}$ under PPT-preserving channels and similar techniques used in the proof of~\cite[Proposition~7.84]{KW20}.
\end{IEEEproof}

Next, by using~\cref{prop:PPT_measured_joint_covariance} and~\cref{prop:Measured_HS_Isotro}, we derive an analytical expression for $E_\gamma^{\operatorname{PPT}}$ for two depolarizing channels.

\begin{proposition}[Depolarizing Channels] \label{prop:channel_div_PPT_depol}
    Let $p,q \in [0,1]$ and $\gamma \geq 1$. Then 
$E_\gamma^{\operatorname{PPT}}\!\left( \cA_{\operatorname{Dep}}^q \middle \Vert \cA_{\operatorname{Dep}}^p \right)=  \max\left\{0,  1-q -\gamma (1-p) + \frac{(q-\gamma p)}{d}, \frac{d-1}{d} \left(q -\gamma p \right) \right\}$, 
    where $\cA_{\operatorname{Dep}}^p$ is a depolarizing channel with parameter $p$ 
    and $d$ is the dimension of the input space of the channel $\cA_{\operatorname{Dep}}^p$.  
\end{proposition}
\begin{IEEEproof}
 The proof follows by applying~\cref{prop:PPT_measured_joint_covariance} together with the analytical expressions for $E_\gamma^{\operatorname{PPT}}$ of isotropic states (\cref{Rem:Isotropic_E_g}). See~\cref{App:channel_div_PPT_depol_proof}.
\end{IEEEproof}
\section{Concluding Remarks and Future Work}

In this paper, we provided a comprehensive study of measured hockey-stick divergences for states and channels. We showed that they satisfy several desirable properties, 
including data processing, triangular inequality, and convexity. We also showed that, for some classes of measurements and states, they are SDP computable or analytically characterized. We discussed their use in providing privacy for states and channels, with an emphasis on quantum pufferfish privacy.

Future work includes finding quantum algorithms to estimate measured hockey-stick divergences. 
Another interesting future research direction is to obtain SDP-computable bounds for locally measured hockey-stick divergences, by using $k$-extendible measurements~\cite{singh2024extendible}. Such a characterization will find use in privacy auditing 
when LO, $\operatorname{LO}^\star$, or 1W-LOCC operations constitute the set of allowed measurements.

\medskip
\textbf{Acknowledgment}: We acknowledge support from the National Science Foundation
under Grant No.~2329662.
\bibliographystyle{ieeetr}
\bibliography{reference}

\begin{thebibliography}{10}

\bibitem{donald1986relative}
M.~J. Donald, ``On the relative entropy,'' {\em Communications in Mathematical Physics}, vol.~105, pp.~13--34, 1986.

\bibitem{hiai1991proper}
F.~Hiai and D.~Petz, ``The proper formula for relative entropy and its asymptotics in quantum probability,'' {\em Communications in Mathematical Physics}, vol.~143, pp.~99--114, 1991.

\bibitem{P09}
M.~Piani, ``Relative entropy of entanglement and restricted measurements,'' {\em Physical Review Letters}, vol.~103, p.~160504, Oct. 2009.
\newblock arXiv:0904.2705 [quant-ph].

\bibitem{matthews2009distinguishability}
W.~Matthews, S.~Wehner, and A.~Winter, ``Distinguishability of quantum states under restricted families of measurements with an application to quantum data hiding,'' {\em Communications in Mathematical Physics}, vol.~291, pp.~813--843, 2009.

\bibitem{cheng2023discrimination}
H.-C. Cheng, A.~Winter, and N.~Yu, ``Discrimination of quantum states under locality constraints in the many-copy setting,'' {\em Communications in Mathematical Physics}, vol.~404, no.~1, pp.~151--183, 2023.

\bibitem{rippchen2024locally}
T.~Rippchen, S.~Sreekumar, and M.~Berta, ``Locally-measured {R}\'enyi divergences,'' 2024.
\newblock arXiv:2405.05037.

\bibitem{asoodeh2021local}
S.~Asoodeh, M.~Aliakbarpour, and F.~P. Calmon, ``Local differential privacy is equivalent to contraction of ${E}_\gamma$-divergence,'' 2021.
\newblock arXiv:2102.01258.

\bibitem{hirche2023quantum}
C.~Hirche, C.~Rouz{\'e}, and D.~S. Fran{\c{c}}a, ``Quantum differential privacy: An information theory perspective,'' {\em IEEE Transactions on Information Theory}, vol.~69, no.~9, pp.~5771--5787, 2023.
\newblock arXiv:2202.10717.

\bibitem{angrisani2023unifying}
A.~Angrisani, M.~Doosti, and E.~Kashefi, ``A unifying framework for differentially private quantum algorithms,'' 2023.
\newblock arXiv:2307.04733.

\bibitem{nuradha_QPP}
T.~Nuradha, Z.~Goldfeld, and M.~M. Wilde, ``Quantum pufferfish privacy: A flexible privacy framework for quantum systems,'' {\em IEEE Transactions on Information Theory}, vol.~70, no.~8, pp.~5731--5762, 2024.

\bibitem{hirche2023quantum_2}
C.~Hirche and M.~Tomamichel, ``Quantum {R}\'enyi and $f$-divergences from integral representations,'' {\em Communications in Mathematical Physics}, vol.~405, no.~9, p.~208, 2024.
\newblock arXiv:2306.12343.

\bibitem{QDP_computation17}
L.~Zhou and M.~Ying, ``Differential privacy in quantum computation,'' in {\em Proceedings of IEEE Computer Security Foundations Symposium (CSF)}, pp.~249--262, IEEE, 2017.

\bibitem{KM14}
D.~{Kifer} and A.~{Machanavajjhala}, ``Pufferfish: A framework for mathematical privacy definitions,'' {\em ACM Transactions on Database Systems}, vol.~39, no.~1, pp.~1--36, 2014.

\bibitem{nuradha2022pufferfishJ}
T.~Nuradha and Z.~Goldfeld, ``Pufferfish privacy: An information-theoretic study,'' {\em IEEE Transactions on Information Theory}, vol.~69, no.~11, pp.~7336--7356, 2023.

\bibitem{nuradha_MIPP}
T.~Nuradha and Z.~Goldfeld, ``An information-theoretic characterization of pufferfish privacy,'' in {\em 2022 IEEE International Symposium on Information Theory (ISIT)}, pp.~2005--2010, 2022.

\bibitem{sharma2012strong}
N.~Sharma and N.~A. Warsi, ``On the strong converses for the quantum channel capacity theorems.'' arXiv:1205.1712, 2012.

\bibitem{nuradha2024contraction}
T.~Nuradha and M.~M. Wilde, ``Contraction of private quantum channels and private quantum hypothesis testing,'' {\em Accpeted to IEEE Transactions on Information Theory}, 2025.
\newblock arXiv:2406.18651.

\bibitem{Wer89}
R.~F. Werner, ``Quantum states with {E}instein-{P}odolsky-{R}osen correlations admitting a hidden-variable model,'' {\em Physical Review A}, vol.~40, pp.~4277--4281, Oct. 1989.

\bibitem{Christoph2024sample}
H.-C. Cheng, C.~Hirche, and C.~Rouz{\'e}, ``Sample complexity of locally differentially private quantum hypothesis testing,'' 2024.
\newblock arXiv:2406.18658.

\bibitem{Pirandola2011}
S.~Pirandola, ``Quantum reading of a classical digital memory,'' {\em Physical Review Lett.}, vol.~106, p.~090504, Mar. 2011.

\bibitem{chiribella2008transforming}
G.~Chiribella, G.~M. D'Ariano, and P.~Perinotti, ``Transforming quantum operations: Quantum supermaps,'' {\em Europhysics Letters}, vol.~83, no.~3, p.~30004, 2008.

\bibitem{gour2019comparison}
G.~Gour, ``Comparison of quantum channels by superchannels,'' {\em IEEE Transactions on Information Theory}, vol.~65, no.~9, pp.~5880--5904, 2019.

\bibitem{KW20}
S.~Khatri and M.~M. Wilde, ``Principles of quantum communication theory: A modern approach,'' 2020.
\newblock arXiv:2011.04672v2.

\bibitem{singh2024extendible}
V.~Singh, T.~Nuradha, and M.~M. Wilde, ``Extendible quantum measurements and limitations on classical communication,'' 2024.
\newblock arXiv:2412.18556.

\bibitem{HH99}
M.~Horodecki and P.~Horodecki, ``Reduction criterion of separability and limits for a class of distillation protocols,'' {\em Physical Review A}, vol.~59, pp.~4206--4216, June 1999.
\newblock arXiv:quant-ph/9708015.

\end{thebibliography}

\onecolumn
 \appendix

 \subsection{Other Properties of Measured Hockey-Stick Divergence}

 \begin{lemma}\label{lem:hockey_stick_gamma_leq}
    Let $\gamma \geq 0$, and let $\rho$ and $\sigma$ be states. 
    The following equality holds: \begin{equation}
        E_\gamma^\mathcal{M}(\rho \Vert \sigma) = \gamma E_{\frac{1}{\gamma}}^\mathcal{M}(\sigma \Vert \rho),
    \end{equation}
where $E_\gamma^\mathcal{M}(\cdot \Vert \cdot)$ is defined in~\eqref{eq:hockey_stick_alter}.
\end{lemma}
\begin{IEEEproof}
Let $0 \leq \gamma \leq 1$.
    By applying~\eqref{eq:hockey_stick_define}, consider that 
    \begin{align}
         E_\gamma^\mathcal{M}(\rho \Vert \sigma) &= \sup_{ M\in \mathcal{M}_2} \Tr\!\left[M(\rho -\gamma \sigma) \right] - (1-\gamma) \\
         &= \sup_{ M\in \mathcal{M}_2} \Tr\!\left[(I-M)(\rho -\gamma \sigma )\right] - (1-\gamma) \\
         &=  \sup_{ M\in \mathcal{M}_2} \Tr\!\left[M(\gamma \sigma -\rho) \right] \\
         &= \gamma \sup_{ M \in \mathcal{M}_2}  \Tr\!\left[M \left( \sigma -\frac{1}{\gamma}\rho \right) \right] \\
         &= \gamma E_{\frac{1}{\gamma}}^\mathcal{M}(\sigma \Vert \rho) \label{eq:relation_gamma_leq_1},
    \end{align}
    where the last inequality follows by applying~\eqref{eq:hockey_stick_define} with $1/\gamma \geq 1$ since $\gamma \leq 1$.

    For $\gamma' \geq 1$, by substituting $\gamma =1/\gamma'$ in~\eqref{eq:relation_gamma_leq_1}, we conclude the proof.
\end{IEEEproof}

 \begin{lemma} \label{lem:E_1_gamma}
    Let $\rho$ and $\sigma$ be states, and let $\gamma \geq 1$. The following inequality holds:
    \begin{equation}
     (\gamma +1)   E_1^\cM(\rho \Vert \sigma) \leq E_\gamma^\cM(\rho \Vert \sigma) + E_\gamma^\cM(\sigma \Vert \rho) + \gamma -1.
    \end{equation}
\end{lemma}
\begin{proof}
   By applying~\eqref{eq:hockey_stick_define}, consider that 
    \begin{align}
        (\gamma+1) E_1^\cM(\rho \Vert \sigma) &=  (\gamma+1) \sup_{M\in \cM_2} \Tr\! \left[ M (\rho- \sigma)\right]\\ 
        &= \sup_{M \in \cM_2} \left\{\Tr[M(\rho-\gamma \sigma)] + \Tr[M(\gamma\rho- \sigma)]\right\} \\ 
        &\leq \sup_{M\in \cM_2} \left\{\Tr[M(\rho-\gamma \sigma)]\right\} + \sup_{M\in \cM_2} \left\{\Tr[M(\gamma\rho- \sigma)] \right\} \\ 
         &= E_\gamma^\cM (\rho \Vert \sigma) + \gamma \sup_{M \in \cM_2} \Tr\!\left[M\left(\rho- \frac{1}{\gamma}\sigma\right)\right] \\ 
         & = E_\gamma^{\cM}(\rho \Vert \sigma) + \gamma \left(E_{\frac{1}{\gamma}}^\cM(\rho \Vert \sigma) +1 -\frac{1}{\gamma} \right)\\
        &= E_\gamma^\cM(\rho \Vert\sigma) +E_\gamma^\cM(\sigma \Vert \rho) + \gamma -1,
    \end{align}
where the penultimate equality follows from~\eqref{eq:hockey_stick_define} with $1/\gamma \leq 1$ and the last equality follows by some algebraic manipulations, along with~\cref{lem:hockey_stick_gamma_leq}.
   
\end{proof}
 \subsection{Proof of \cref{prop:equi_hockey_stick}} \label{App:eqqui_hockey_stick_proof}
Consider an arbitrary POVM $\left\{M,I-M\right\} \in \mathcal{M}_2$, where $\mathcal{M}_2$ is defined in~\eqref{eq:M_2_set}. Then the following inequality holds for every $\gamma \ge 1$:
\begin{align}
\operatorname{Tr}\!\left[ M(\rho -\gamma \sigma )\right] 
    &\leq \max\left\{0, \operatorname{Tr}\!\left[ M(\rho -\gamma \sigma )\right] \right\} \\ 
    & \leq  \max\left\{0, \operatorname{Tr}\!\left[ M(\rho -\gamma \sigma )\right] \right\}+  \max\left\{0, \operatorname{Tr}\!\left[ (I-M)(\rho -\gamma \sigma )\right] \right\} \\ 
     &\leq \sup_{\{M_x \in \mathcal{M}\}_x}  \sum_{x} \max \left\{0,\operatorname{Tr}\!\left[ M_x(\rho -\gamma \sigma )\right] \right \},
\end{align}
where the last inequality follows because $\{M, I-M\}$ is a POVM in $\mathcal{M}$. By optimizing the left-hand side  over all $M \in \mathcal{M}_2$, we arrive at 
\begin{equation}
\label{eq:one_side_bound}
     E_\gamma^{\mathcal{M}}(\rho \Vert \sigma) \leq \widehat{E}_\gamma^{\mathcal{M}}(\rho \Vert \sigma).
\end{equation}

To prove the reverse inequality, we consider a POVM $\{M_x \in \mathcal{M}\}_{x\in \mathcal{X}}$ that obeys the coarse-graining condition; that is, $M_x + M_{x'} \in \mathcal{M}$ for all $x,x' \in \mathcal{X}$. Consider the following equality:
\begin{align}
   \sum_{x} \max \left\{0,\operatorname{Tr}\!\left[ M_x(\rho -\gamma \sigma )\right] \right \}  
   &= \sum_{x:\operatorname{Tr}\left[ M_x(\rho -\gamma \sigma )\right] \geq 0 } \operatorname{Tr}\!\left[ M_x(\rho -\gamma \sigma )\right] \\ 
   &=\operatorname{Tr}\!\left[ M_+(\rho -\gamma \sigma )\right] \\
   & \leq\sup_{M \in \mathcal{M}_2} \operatorname{Tr}\!\left[ M(\rho -\gamma \sigma )\right],
\end{align}
where the second equality follows by defining 
\begin{equation}
    M_+ \coloneqq \sum_{x:\operatorname{Tr}\left[ M_x(\rho -\gamma \sigma )\right] \geq 0 } M_x
\end{equation}
and the last inequality follows because $M_+, I- M_+\in \mathcal{M}$ with the coarse-graining assumption and 
\begin{equation}
    I-M_+ =  \sum_{x:\operatorname{Tr}\left[ M_x(\rho -\gamma \sigma )\right] < 0 } M_x.
\end{equation}
Finally, optimizing over all POVMs $\{M_x \in \mathcal{M}\}_x$, we obtain the inequality
\begin{equation}
 \widehat{E}_\gamma^{\mathcal{M}}(\rho \Vert \sigma) \leq      E_\gamma^{\mathcal{M}}(\rho \Vert \sigma),
\end{equation}
and together with~\eqref{eq:one_side_bound} we conclude the proof.

\subsection{Proof of \cref{prop:properties_m_HS}} 

\label{App:properties_proof}

    \underline{Data Processing:}
    Under the assumption that the channel $\mathcal{N}$ is $\mathcal{M}$-compatible, it follows that  $\cN^\dag (M) \in \cM$ for all $M \in \cM$. Since the adjoint of a trace-preserving map is unital, it also follows that $\cN^\dag(I)=I$, leading to $ \cN^\dag(I-M)= I- \cN^\dag (M) \in \cM$. With that and fixing $M\in \cM_2$ such that $M, I-M \in \cM$,
 consider that 
    \begin{align}
        \Tr\!\left[M \left(\cN(\rho) -\gamma \cN(\sigma)\right) \right] & = \Tr\!\left[ \cN^\dag(M) \left(\rho-\gamma \sigma\right) \right] \\ 
        & \leq \sup_{M'\in \cM_2} \Tr\!\left[ M' \left(\rho-\gamma \sigma\right) \right]
        \\& =E_\gamma^\cM(\rho \Vert \sigma).
    \end{align}
We arrive at the desired inequality by optimizing the left-hand side over all $M \in \cM$ such that $I-M \in \cM$.

\underline{Triangular Inequality:}
Let $\gamma_1,\gamma_2 \ge 1$. Then
\begin{align}  E_{\gamma_1 \gamma_2}^\cM (\rho \Vert \sigma)
   &= \sup_{M \in \cM_2}  \Tr\!\left[ M \left(\rho-\gamma_1 \gamma_2 \sigma\right) \right] \\ 
   &= \sup_{M\in \cM_2}  \Tr\!\left[ M \left(\rho- \gamma_1 \tau + \gamma_1 \tau -\gamma_1 \gamma_2 \sigma\right) \right]  \\ 
   & \leq \sup_{M \in \cM_2}  \Tr\!\left[ M \left(\rho-\gamma_1 \tau\right) \right] + \sup_{M\in \cM_2}  \Tr\!\left[ M \left(\gamma_1 \tau-\gamma_1 \gamma_2 \sigma\right) \right] \\
   &=E_{\gamma_1}^\cM(\rho \Vert \tau) + \gamma_1 E_{\gamma_2}^\cM(\tau \Vert \sigma).
\end{align}

\underline{Monotonicity:} 
Let $\gamma_1\ge \gamma_2\ge 1$. Then $\gamma_1 \sigma \geq \gamma_2 \sigma$ and $\rho -\gamma_1 \sigma \leq \rho-\gamma_2 \sigma$. 
Since $M \geq 0$, we have that for all $M \in \cM_2$
\begin{equation}
    \Tr\!\left[ M \left(\rho-\gamma_1 \sigma\right) \right] \leq \Tr\!\left[ M \left(\rho-\gamma_2 \sigma\right) \right].
\end{equation}
Supremizing over all $M \in \cM_2$, we obtain the desired inequality.

\underline{Convexity:} 
Let $\left\{p_x\right\}_{x\in \mathcal{X}}$ be a probability distribution, and let $\left\{\rho_{x}\right\}_{x\in \mathcal{X}}$ and $\left\{\sigma_{x}\right\}_{x\in \mathcal{X}}$ be sets of quantum states. Then for every $M\in \cM_2$ and $\gamma\ge 1$, the following equality holds: 
\begin{align}
      \Tr\!\left[ M \left(\rho-\gamma \sigma\right) \right] &= \sum_x p_x \Tr\!\left[ M \left(\rho_x-\gamma \sigma_x\right) \right] \\
      & \leq \sum_x p_x 
 \sup_{M \in \cM_2} \Tr\!\left[ M \left(\rho_x-\gamma \sigma_x\right) \right] \\
 &= \sum_x p_x E_\gamma^\cM(\rho_x \Vert \sigma_x).
\end{align}
By optimizing over $M \in \cM_2$ on the left-hand side, we conclude the proof.

\subsection{Proof of \cref{prop:PPT_Measured_HS_SDP}}

\label{app:PPT_Measured_HS_SDP}

Recall the standard form of primal and dual SDPs, as characterized by the Hermitian matrices $A$ and $B$ and a Hermiticity-preserving superoperator $\Phi$ \cite[Definition~2.26]{KW20}:
\begin{align}
 \sup_{X\geq0}\left\{  \operatorname{Tr}[AX]:\Phi(X)\leq B\right\} ,\\
 \inf_{Y\geq0}\left\{  \operatorname{Tr}[BY]:\Phi^{\dag}(Y)\geq A\right\}  .
\end{align}
The SDP for the PPT measured hockey-stick divergence can be written in the standard form as follows: 
\begin{align}
A&= \rho- \lambda \sigma, \quad X=M, \\
B  &  =
\begin{bmatrix}
0 & 0 & 0 & 0\\
0 & 0 & 0 & 0\\
0 & 0 & I_{AB} & 0 \\
0 & 0 & 0 & I_{AB}
\end{bmatrix}
,\\
\Phi(X)  &  =
\begin{bmatrix}
-M & 0 & 0 & 0 \\
0 & -T_B(M) & 0 & 0 \\
0 & 0 & M & 0 \\
0  & 0 & 0 & T_B(M)
\end{bmatrix}.
\end{align}
Setting
\begin{equation}
Y=
\begin{bmatrix}
Y_1 & 0 & 0 & 0 \\
0  & Y_2 & 0 & 0 \\
0  & 0 & Y_3  & 0\\
0 & 0 & 0 & Y_4
\end{bmatrix}
,
\end{equation}
we find that
\begin{align}
   \Tr\!\left[ \Phi(X) Y \right] 
   &= \Tr\!\left[ -M Y_1 - T_B(M) Y_2 + M Y_3 +T_B(M) Y_4\right] \\
    &=\Tr\!\left[ M (Y_3 -Y_1) + T_B(M) (Y_4 -  Y_2)\right] \\ 
    &=\Tr\!\left[ M \left(Y_3- Y_1 +T_B(Y_4 -Y_2)\right)\right],
\end{align}
where the last equality holds due to $\Tr[T_B(L_{AB}) S_{AB}]= \Tr\!\left[ L_{AB} T_B(S_{AB}) \right]$ with the adjoint of the partial transpose superoperator being the partial transpose superoperator.
This leads to
\begin{equation}
    \Phi^\dag(Y) =Y_3- Y_1 +T_B(Y_4 -Y_2).
\end{equation}

Strong duality holds due to Slater's conditions by the following choice of feasible and strictly feasible solutions: choose $Y_3= (\rho -\gamma \sigma)_+$, $Y_i=0$ for all $i \in \{1,2,4\}$ as a feasible solution for the dual SDP and choose $M= (1-\delta) I_{AB}$ with $\delta \in (0,1)$ as a strictly feasible solution to the primal SDP.
Together with the strong duality, we conclude the proof.

\subsection{Werner States} \label{App:Werner_States}
Here we analyze the measured hockey-stick divergence between two Werner states. We first obtain a simpler expression for the measured hockey-stick divergence between two Werner states using the symmetries of Werner states, which we state in Lemma~\ref{lem:meas_hs_wer_symm}. We then use the statement of Lemma~\ref{lem:meas_hs_wer_symm} to prove Propositions~\ref{prop:HS_Werner_general} and~\ref{prop:Measured_HS_Werner}.

\medskip

Let $\Pi^{\operatorname{sym}}_{AB}$ and $\Pi^{\operatorname{asym}}_{AB}$ be the projections onto the symmetric and antisymmetric subspaces, respectively. These projections can be written in terms of the identity and swap operator as follows:
\begin{align}
    \Pi^{\operatorname{sym}}_{AB} &\coloneqq \frac{1}{2}\!\left(I_{AB}+F_{AB}\right),\\
    \Pi^{\operatorname{asym}}_{AB} &\coloneqq \frac{1}{2}\!\left(I_{AB}-F_{AB}\right).
\end{align}
Note that the states $\Theta_{AB}$ and $\Theta^{\perp}_{AB}$ can be obtained by normalizing the aforementioned projectors. That is,
\begin{align}
    \Theta_{AB} &= \frac{\Pi^{\operatorname{sym}}_{AB}}{\operatorname{Tr}\!\left[\Pi^{\operatorname{sym}}_{AB}\right]} = \frac{2}{d(d+1)}\Pi^{\operatorname{sym}}_{AB},\label{eq:symm_proj_state}\\
    \Theta^{\perp}_{AB} &= \frac{\Pi^{\operatorname{asym}}_{AB}}{\operatorname{Tr}\!\left[\Pi^{\operatorname{asym}}_{AB}\right]} = \frac{2}{d(d-1)}\Pi^{\operatorname{asym}}_{AB}.\label{eq:asym_proj_state}
\end{align}

\begin{lemma}\label{lem:meas_hs_wer_symm}
    Fix $p,q \in [0,1]$. For $\gamma \ge 1$, the measured hockey-stick divergence between two Werner states, $\omega^q_{AB}$ and $\omega^p_{AB}$, is equal to the following:
    \begin{equation}\label{eq:meas_hs_wer_symm}
        E^{\mathcal{M}}_{\gamma}\!\left(\omega^q_{AB}\Vert\omega^p_{AB}\right) = \sup_{M \in \mathcal{M}_2} \left[ \frac{\operatorname{Tr}\!\left[M_{AB}\Pi^{\operatorname{sym}}_{AB}\right]}{\operatorname{Tr}\!\left[\Pi^{\operatorname{sym}}_{AB}\right]}(q-\gamma p) + \frac{\operatorname{Tr}\!\left[M_{AB}\Pi^{\operatorname{asym}}_{AB}\right]}{\operatorname{Tr}\!\left[\Pi^{\operatorname{asym}}_{AB}\right]}(1-q-\gamma(1-p))\right],
    \end{equation}
    where $\Pi^{\operatorname{sym}}_{AB}$ and $\Pi^{\operatorname{asym}}_{AB}$  are the projections onto the symmetric and antisymmetric subspaces, respectively.
\end{lemma}
\begin{proof}
Consider the following twirling channel:
\begin{equation}
    \mathcal{T}_{AB}\!\left(\cdot\right) \coloneqq \int dU \left(U_A\otimes U_B\right)\!\left(\cdot\right)\left(U_A\otimes U_B\right)^{\dagger}.
\end{equation}
The action of this channel on an operator $X_{AB} \in \mathcal{L}(\mathcal{H_{A}\otimes \mathcal{H}_B})$ results in an operator of the following form~\cite{Wer89}:
\begin{equation}\label{eq:twirled_werner_op}
    \mathcal{T}_{AB}\!\left(X_{AB}\right) = \frac{\operatorname{Tr}\!\left[X_{AB}\Pi^{\operatorname{sym}}_{AB}\right]}{\operatorname{Tr}\!\left[\Pi^{\operatorname{sym}}_{AB}\right]}\Pi^{\operatorname{sym}}_{AB} + \frac{\operatorname{Tr}\!\left[X_{AB}\Pi^{\operatorname{asym}}_{AB}\right]}{\operatorname{Tr}\!\left[\Pi^{\operatorname{asym}}_{AB}\right]}\Pi^{\operatorname{asym}}_{AB}.
\end{equation}
It is easily verified that the Werner states are invariant under the action of $\mathcal{T}_{AB}$. Also, note that  the Hilbert--Schmidt adjoint of $\mathcal{T}_{AB}$ is $\mathcal{T}_{AB}$ itself. 

The measured hockey-stick divergence between two Werner states, $\omega^q_{AB}$ and $\omega^p_{AB}$, can be written as follows:
\begin{align}
    E^{\mathcal{M}}_{\gamma}\!\left(\omega^q_{AB}\Vert\omega^p_{AB}\right) &= \sup_{M \in \mathcal{M}_2} \operatorname{Tr}\!\left[M_{AB}\!\left(\omega^q_{AB} - \gamma \omega^p_{AB}\right)\right]\\
    &= \sup_{M \in \mathcal{M}_2} \operatorname{Tr}\!\left[M_{AB}\!\left(\mathcal{T}_{AB}\!\left(\omega^q_{AB}\right) - \gamma \mathcal{T}_{AB}\!\left(\omega^p_{AB}\right)\right)\right]\\
    &= \sup_{M \in \mathcal{M}_2} \operatorname{Tr}\!\left[\mathcal{T}_{AB}\!\left(M_{AB}\right)\!\left(\omega^q_{AB} - \gamma \omega^p_{AB}\right)\right],\label{eq:twirl_measurement_opt_werner}\\
    &= \sup_{M \in \mathcal{M}_2} \frac{\operatorname{Tr}\!\left[M_{AB}\Pi^{\operatorname{sym}}_{AB}\right]}{\operatorname{Tr}\!\left[\Pi^{\operatorname{sym}}_{AB}\right]}\operatorname{Tr}\!\left[\Pi^{\operatorname{sym}}_{AB}\!\left(\omega^q_{AB} - \gamma \omega^p_{AB}\right)\right] + \frac{\operatorname{Tr}\!\left[M_{AB}\Pi^{\operatorname{asym}}_{AB}\right]}{\operatorname{Tr}\!\left[\Pi^{\operatorname{asym}}_{AB}\right]}\operatorname{Tr}\!\left[\Pi^{\operatorname{asym}}_{AB}\!\left(\omega^q_{AB} - \gamma \omega^p_{AB}\right)\right]\\
    &= \sup_{M \in \mathcal{M}_2} \frac{\operatorname{Tr}\!\left[M_{AB}\Pi^{\operatorname{sym}}_{AB}\right]}{\operatorname{Tr}\!\left[\Pi^{\operatorname{sym}}_{AB}\right]}(q-\gamma p) + \frac{\operatorname{Tr}\!\left[M_{AB}\Pi^{\operatorname{asym}}_{AB}\right]}{\operatorname{Tr}\!\left[\Pi^{\operatorname{asym}}_{AB}\right]}(1-q-\gamma(1-p)),
\end{align}
where the second equality follows from the fact that Werner states are invariant under the action of the twirling channel $\mathcal{T}_{AB}$, the third equality follows from the fact that $\mathcal{T}_{AB}$ is self-adjoint, the penultimate equality follows from~\eqref{eq:twirled_werner_op}, and the ultimate equality follows from the definition of Werner states along with~\eqref{eq:symm_proj_state} and~\eqref{eq:asym_proj_state}.
\end{proof}

\medskip

\noindent \textbf{Proof of~\cref{prop:HS_Werner_general}:}
The statement of Proposition~\ref{prop:HS_Werner_general} can be obtained from Lemma~\ref{lem:meas_hs_wer_symm} by considering $\mathcal{M}$ to be the set of all measurements. 

For $\gamma \ge 1$, at most one of the two quantities can be positive: $q-\gamma p$ or $1-q-\gamma (1-p)$. Choosing $\mathcal{M}$ to be the set of all measurements and ignoring the negative term in~\eqref{eq:meas_hs_wer_symm}, we arrive at the following:
\begin{align}
    E_{\gamma}\!\left(\omega^q_{AB}\Vert\omega^p_{AB}\right) &= \sup_{0 \le M \le I} \left[\frac{\operatorname{Tr}\!\left[M_{AB}\Pi^{\operatorname{sym}}_{AB}\right]}{\operatorname{Tr}\!\left[\Pi^{\operatorname{sym}}_{AB}\right]}(q-\gamma p) + \frac{\operatorname{Tr}\!\left[M_{AB}\Pi^{\operatorname{asym}}_{AB}\right]}{\operatorname{Tr}\!\left[\Pi^{\operatorname{asym}}_{AB}\right]}(1-q-\gamma(1-p))\label{eq:hs_werner_symm}\right]\\
    &\le \sup_{0\le M\le I}\max\left\{\frac{\operatorname{Tr}\!\left[M_{AB}\Pi^{\operatorname{sym}}_{AB}\right]}{\operatorname{Tr}\!\left[\Pi^{\operatorname{sym}}_{AB}\right]}(q-\gamma p), \frac{\operatorname{Tr}\!\left[M_{AB}\Pi^{\operatorname{asym}}_{AB}\right]}{\operatorname{Tr}\!\left[\Pi^{\operatorname{asym}}_{AB}\right]}(1-q-\gamma(1-p))\right\}\\
    &\le \max\left\{0, q-\gamma p, 1-q-\gamma(1-p)\right\},\label{eq:hs_werner_ub}
\end{align}
where the last inequality follows from the fact that $M\le I$ and that $M = 0$ is a valid choice of measurement.

To show that the inequality in~\eqref{eq:hs_werner_ub} is saturated, consider $\Pi^{\operatorname{sym}}_{AB}$ to be a specific choice for the measurement operator in~\eqref{eq:hs_werner_symm}, which implies $E_{\gamma}\!\left(\omega^q_{AB}\Vert\omega^p_{AB}\right)\ge q-\gamma p$. Furthermore, choosing $M_{AB} = \Pi^{\operatorname{asym}}_{AB}$ leads to $E_{\gamma}\!\left(\omega^q_{AB}\Vert\omega^p_{AB}\right)\ge 1-q-\gamma(1-p)$, and choosing $M_{AB} = 0$ leads to $E_{\gamma}\!\left(\omega^q_{AB}\Vert\omega^p_{AB}\right)\ge 0$. Combining the three inequalities, we arrive at the following inequality:
\begin{equation}
    E_{\gamma}\!\left(\omega^q_{AB}\Vert\omega^p_{AB}\right)\ge \max\left\{0,q-\gamma p,1-q- \gamma(1-p)\right\},
\end{equation}
which completes the proof.

\bigskip
\noindent \textbf{Proof of~\cref{prop:Measured_HS_Werner}:}
\underline{Lower bound:} 
Recall that for $\gamma \geq 1$
\begin{equation}
    E_\gamma^{\operatorname{LO}^\star}(\rho \Vert \sigma) \coloneqq \sup_{M \in \operatorname{LO}^\star} \Tr\!\left[ M (\rho -\gamma \sigma)\right].
\end{equation}
Since the measurement operators $M=0$ and $I-M =I$, trivially belong to the set of $\operatorname{LO}^\star$ operators, we have
\begin{equation}
    E_\gamma^{\operatorname{LO}^\star}(\omega^q \Vert \omega^p) \geq 0 \label{eq:M_o_bound_2_W}.
\end{equation}

Observe that $M= \sum_{i=1}^d |i \rangle\!\langle i| \otimes |i \rangle\!\langle i|$  and $I-M= \sum_{i\neq j} |i \rangle\!\langle i | \otimes |j\rangle \! \langle j| =I- \sum_{i=1}^d |i \rangle\!\langle i| \otimes |i \rangle\!\langle i|$ belong to the set of $\operatorname{LO}^\star$ operators as we can see by the following argument: This measurement can be implemented by Alice and Bob first measuring their local systems in the computational basis. Then they perform a classical post-processing of their measurement outcomes by accepting if the outcomes match, which corresponds to the measurement operator $\sum_{i=1}^d|i\rangle\!\langle i|_A\otimes |i\rangle\!\langle i|_A$. and rejecting if the outcomes do not match, which corresponds to the measurement operator $I_{AB}-\sum_{i=1}^d|i\rangle\!\langle i|_A\otimes |i\rangle\!\langle i|_A$.
With the former, we obtain the inequality 
\begin{align}
    E_\gamma^{\operatorname{LO}^\star}(\omega^q \Vert \omega^p) &\geq \Tr\!\left[\sum_{i=1}^d |i \rangle\!\langle i| \otimes |i \rangle\!\langle i| \left( (q-\gamma p) \Theta + \left( (1-q)- \gamma (1-p) \right) \Theta^\perp \right) \right] \\ 
    &=  (q-\gamma p) \times \frac{2}{d+1} + \left( (1-q)- \gamma (1-p) \right)  \times 0 \\
    &= \frac{2 (q-\gamma p)}{d+1}, \label{eq:one_M_bound_2_W}
\end{align}
where the first equality followed by 
$\Tr[M \Theta]= 2/(d+1)$ and $\Tr[M \Theta^\perp]=0$.

Furthermore, with the choice $I- M= \sum_{i\neq j} |i \rangle\!\langle i | \otimes |j\rangle \! \langle j|$, we arrive at another lower bound as follows: 
\begin{align}
    E_\gamma^{\operatorname{LO}^\star}(\omega^q \Vert \omega^p) &\geq \Tr\!\left[\left( I- \sum_{i=1}^d |i \rangle\!\langle i| \otimes |i \rangle\!\langle i| \right)\left( (q-\gamma p) \Theta + \left( (1-q)- \gamma (1-p) \right) \Theta^\perp \right) \right] \\ 
    &= 1- \gamma -\frac{2 (q-\gamma p)}{d+1}\label{eq:one_M_bound_3_W},
\end{align}
where the equality follows from~\eqref{eq:one_M_bound_2_W}.

Combining~\eqref{eq:M_o_bound_2_W}, \eqref{eq:one_M_bound_2_W}, and \eqref{eq:one_M_bound_3_W}, we obtain the following inequality:
\begin{equation} \label{eq:LO_bound_2_1_W}
     E_\gamma^{\operatorname{LO}^\star}(\omega^q \Vert \omega^p) \geq  \max\left\{0, \frac{2(q-\gamma p)} {(d+1)}, (1-\gamma)- \frac{2(q-\gamma p)} {(d+1)}\right\}.
\end{equation}

\medskip 
\underline{Upper bound:} Now we will show that the lower bound on $E_\gamma^{\operatorname{LO}^\star}(\omega^q \Vert \omega^p)$ obtained in~\eqref{eq:LO_bound_2_1_W} is also an upper bound on $E_\gamma^{\operatorname{PPT}}(\omega^q \Vert \omega^p)$.

  Let $M_{AB}$ be a PPT measurement operator; that is, $0\le M_{AB} \le I_{AB}$ and $0\le T_B\!\left(M_{AB}\right) \le I_{AB}$. From Lemma~\ref{lem:meas_hs_wer_symm}, we know that
    \begin{equation}\label{eq:PPT_hs_wer_symm_1}
        E^{\operatorname{PPT}}_{\gamma}\!\left(\omega^q_{AB}\Vert\omega^p_{AB}\right) = \sup_{\substack{0\le M\le I\\ 0\le T_B(M)\le I}} \frac{\operatorname{Tr}\!\left[M_{AB}\Pi^{\operatorname{sym}}_{AB}\right]}{\operatorname{Tr}\!\left[\Pi^{\operatorname{sym}}_{AB}\right]}(q-\gamma p) + \frac{\operatorname{Tr}\!\left[M_{AB}\Pi^{\operatorname{asym}}_{AB}\right]}{\operatorname{Tr}\!\left[\Pi^{\operatorname{asym}}_{AB}\right]}(1-q-\gamma(1-p)).
    \end{equation}
    The condition $0\le M_{AB}\le I_{AB}$ implies that 
    \begin{align}
        0 &\le \operatorname{Tr}\!\left[M_{AB}\Pi^{\operatorname{sym}}_{AB}\right] \le \operatorname{Tr}\!\left[\Pi^{\operatorname{sym}}_{AB}\right],\\
        0 &\le \operatorname{Tr}\!\left[M_{AB}\Pi^{\operatorname{asym}}_{AB}\right] \le \operatorname{Tr}\!\left[\Pi^{\operatorname{asym}}_{AB}\right].\label{eq:PPT_asym_max_cond}
    \end{align}

    The partial transpose of the swap operator is the unnormalized maximally entangled vector. That is, $T_B(F_{AB}) = |\Gamma\rangle\!\langle \Gamma|_{AB}$, where 
    \begin{equation}\label{eq:gamma_vec_defn}
        |\Gamma\rangle_{AB} \coloneqq \sum_{i = 0}^{d-1}|i\rangle_A|i\rangle_B,
    \end{equation}
    where $d$ is the dimension of systems $A$ and $B$.
    We will use the notation $\Gamma_{AB}\coloneqq |\Gamma\rangle\!\langle \Gamma|_{AB}$ for conciseness. Now consider the following identity:
    \begin{equation}\label{eq:sym_asym_proj_diff}
        T_B\!\left(\Pi^{\operatorname{sym}}_{AB} - \Pi^{\operatorname{asym}}_{AB}\right) = T_B\!\left(F_{AB}\right) = \Gamma_{AB}.
    \end{equation}
    The PPT condition for $M_{AB}$ implies the following inequality:
    \begin{align}
        0&\le T_B\!\left(M_{AB}\right)\le I_{AB}\\
        \implies 0 &\le \langle \Gamma|T_B\!\left(M_{AB}\right)|\Gamma\rangle \le \langle \Gamma|\Gamma\rangle\\
        \implies 0 &\le \operatorname{Tr}\!\left[T_{B}\!\left(M_{AB}\right)\Gamma_{AB}\right] \le d\\
        \implies 0 &\le \operatorname{Tr}\!\left[T_{B}\!\left(M_{AB}\right)T_B\!\left(\Pi^{\operatorname{sym}}_{AB} - \Pi^{\operatorname{asym}}_{AB}\right)\right] \le d\\
        \implies 0 &\le \operatorname{Tr}\!\left[M_{AB}\!\left(\Pi^{\operatorname{sym}}_{AB} - \Pi^{\operatorname{asym}}_{AB}\right)\right] \le d\label{eq:PPT_cond_symm_1}\\
        \implies \operatorname{Tr}\!\left[M_{AB}\Pi^{\operatorname{asym}}_{AB}\right] &\le \operatorname{Tr}\!\left[M_{AB}\Pi^{\operatorname{sym}}_{AB}\right] \le \operatorname{Tr}\!\left[M_{AB}\Pi^{\operatorname{asym}}_{AB}\right] + d,\label{eq:PPT_cond_symm}
    \end{align}
    where the fourth line follows from~\eqref{eq:sym_asym_proj_diff} and the fifth line follows from the fact that the partial transpose map is the Hilbert--Schmidt adjoint of itself as well as the inverse of itself.

    Let us first assume that $q-\gamma p \ge 0$. We can rewrite~\eqref{eq:PPT_hs_wer_symm_1} as follows:
    \begin{align}
        E^{\operatorname{PPT}}_{\gamma}\!\left(\omega^q_{AB}\Vert\omega^p_{AB}\right) &= \sup_{\substack{0\le M\le I\\ 0\le T_B(M)\le I}} \left(\frac{\operatorname{Tr}\!\left[M_{AB}\Pi^{\operatorname{sym}}_{AB}\right]}{\operatorname{Tr}\!\left[\Pi^{\operatorname{sym}}_{AB}\right]} - \frac{\operatorname{Tr}\!\left[M_{AB}\Pi^{\operatorname{asym}}_{AB}\right]}{\operatorname{Tr}\!\left[\Pi^{\operatorname{asym}}_{AB}\right]}\right)(q-\gamma p) + \frac{\operatorname{Tr}\!\left[M_{AB}\Pi^{\operatorname{asym}}_{AB}\right]}{\operatorname{Tr}\!\left[\Pi^{\operatorname{asym}}_{AB}\right]}(1-\gamma)\label{eq:PPT_hs_wer_q-gp}\\
        &\le \sup_{\substack{0\le M\le I\\ 0\le T_B(M)\le I}} \left(\frac{\operatorname{Tr}\!\left[M_{AB}\Pi^{\operatorname{sym}}_{AB}\right]}{\operatorname{Tr}\!\left[\Pi^{\operatorname{sym}}_{AB}\right]} - \frac{\operatorname{Tr}\!\left[M_{AB}\Pi^{\operatorname{asym}}_{AB}\right]}{\operatorname{Tr}\!\left[\Pi^{\operatorname{sym}}_{AB}\right]}\right)(q-\gamma p) + \frac{\operatorname{Tr}\!\left[M_{AB}\Pi^{\operatorname{asym}}_{AB}\right]}{\operatorname{Tr}\!\left[\Pi^{\operatorname{asym}}_{AB}\right]}(1-\gamma)\\
        &\le \sup_{\substack{0\le M\le I\\ 0\le T_B(M)\le I}} \frac{d}{\operatorname{Tr}\!\left[\Pi^{\operatorname{sym}}_{AB}\right]}(q-\gamma p) + \frac{\operatorname{Tr}\!\left[M_{AB}\Pi^{\operatorname{asym}}_{AB}\right]}{\operatorname{Tr}\!\left[\Pi^{\operatorname{asym}}_{AB}\right]}(1-\gamma)\\
        &\le \frac{2(q-\gamma p)}{d+1},\label{eq:PPT_hs_q_ge_gp}
    \end{align}
    where the first inequality follows by modifying the denominator of the second term because $\operatorname{Tr}\!\left[\Pi^{\operatorname{sym}}_{AB}\right]\ge \operatorname{Tr}\!\left[\Pi^{\operatorname{asym}}_{AB}\right]$, the second inequality follows from~\eqref{eq:PPT_cond_symm_1}, and the final inequality follows by substituting the value of $\operatorname{Tr}\!\left[\Pi^{\operatorname{sym}}_{AB}\right]$ and ignoring the second term in the penultimate inequality because it is negative.

    Now let us assume that $1-q-\gamma(1-p) \ge 0$, which also implies that $q-\gamma p \le 0$ since $\gamma \ge 1$. We can rewrite~\eqref{eq:PPT_hs_wer_q-gp} as follows:
    \begin{align}
        E^{\operatorname{PPT}}_{\gamma}\!\left(\omega^q_{AB}\Vert\omega^p_{AB}\right) &= \sup_{\substack{0\le M\le I\\ 0\le T_B(M)\le I}} \frac{\operatorname{Tr}\!\left[M_{AB}\Pi^{\operatorname{asym}}_{AB}\right]}{\operatorname{Tr}\!\left[\Pi^{\operatorname{asym}}_{AB}\right]}\left(\left(\frac{\operatorname{Tr}\!\left[\Pi^{\operatorname{asym}}_{AB}\right]}{\operatorname{Tr}\!\left[M_{AB}\Pi^{\operatorname{asym}}_{AB}\right]}\frac{\operatorname{Tr}\!\left[M_{AB}\Pi^{\operatorname{sym}}_{AB}\right]}{\operatorname{Tr}\!\left[\Pi^{\operatorname{sym}}_{AB}\right]} - 1\right)(q-\gamma p) + 1-\gamma\right)\\
        &= \sup_{\substack{0\le M\le I\\ 0\le T_B(M)\le I}} \frac{\operatorname{Tr}\!\left[M_{AB}\Pi^{\operatorname{asym}}_{AB}\right]}{\operatorname{Tr}\!\left[\Pi^{\operatorname{asym}}_{AB}\right]}\left(\frac{\operatorname{Tr}\!\left[\Pi^{\operatorname{asym}}_{AB}\right]}{\operatorname{Tr}\!\left[\Pi^{\operatorname{sym}}_{AB}\right]}\frac{\operatorname{Tr}\!\left[M_{AB}\Pi^{\operatorname{sym}}_{AB}\right]}{\operatorname{Tr}\!\left[M_{AB}\Pi^{\operatorname{asym}}_{AB}\right]}(q-\gamma p)  + 1-q-\gamma(1-p)\right)\\
        &\le \sup_{\substack{0\le M\le I\\ 0\le T_B(M)\le I}} \frac{\operatorname{Tr}\!\left[M_{AB}\Pi^{\operatorname{asym}}_{AB}\right]}{\operatorname{Tr}\!\left[\Pi^{\operatorname{asym}}_{AB}\right]}\left(\frac{\operatorname{Tr}\!\left[\Pi^{\operatorname{asym}}_{AB}\right]}{\operatorname{Tr}\!\left[\Pi^{\operatorname{sym}}_{AB}\right]}(q-\gamma p)  + 1-q-\gamma(1-p)\right)\\
        &\le  \max\left\{\frac{\operatorname{Tr}\!\left[\Pi^{\operatorname{asym}}_{AB}\right]}{\operatorname{Tr}\!\left[\Pi^{\operatorname{sym}}_{AB}\right]}(q-\gamma p)  + 1-q-\gamma(1-p),0\right\}\\
        &=  \max\left\{1-\gamma - \frac{2(q-\gamma p)}{d+1},0\right\},\label{eq:PPT_hs_q_le_gp}
    \end{align}
    where the first inequality follows from~\eqref{eq:PPT_cond_symm} and the fact that $q-\gamma p \le 0$, the second inequality follows from~\eqref{eq:PPT_asym_max_cond} and the fact that $M=0$ is a valid choice, and the final equality follows by substituting the values of $\operatorname{Tr}\!\left[\Pi^{\operatorname{sym}}_{AB}\right]$ and $\operatorname{Tr}\!\left[\Pi^{\operatorname{asym}}_{AB}\right]$.

    For the last case, if both $q-\gamma p \le 0$ and $1-q-\gamma(1-p)\le 0$, then the optimal choice is $M_{AB} = 0$ leading to $E^{\operatorname{PPT}_{\gamma}}\!\left(\omega^q_{AB}\Vert\omega^p_{AB}\right) = 0$. Therefore, the inequalities in~\eqref{eq:PPT_hs_q_ge_gp} and~\eqref{eq:PPT_hs_q_le_gp} imply that
    \begin{equation}\label{eq:PPT_bound_2_1_W}
        E^{\operatorname{PPT}}_{\gamma}\!\left(\omega^q_{AB}\Vert\omega^p_{AB}\right) \le \max\left\{0, \frac{2(q-\gamma p)}{d+1}, 1-\gamma - \frac{2(q-\gamma p)}{d+1}\right\}.
    \end{equation}

Recall that
\begin{equation}
    E^{\operatorname{LO}^{\star}}_{\gamma}\!\left(\omega^q_{AB}\Vert\omega^p_{AB}\right) \le E^{\operatorname{1W-LOCC}}_{\gamma}\!\left(\omega^q_{AB}\Vert\omega^p_{AB}\right) \le E^{\operatorname{LOCC}}_{\gamma}\!\left(\omega^q_{AB}\Vert\omega^p_{AB}\right)\le E^{\operatorname{PPT}}_{\gamma}\!\left(\omega^q_{AB}\Vert\omega^p_{AB}\right).
\end{equation}
Therefore, the inequalities in~\eqref{eq:LO_bound_2_1_W} and~\eqref{eq:PPT_bound_2_1_W} imply that all the aforementioned quanitities are equal for $\gamma\ge 1$.

\subsection{Isotropic States}
\label{App:Isotrpic_States}

In this appendix, we analyze the measured hockey-stick divergence between two isotropic states. We first obtain a simpler expression for the measured hockey-stick divergence between two isotropic states using the symmetries of isotropic states, which we state in~\cref{lem:meas_hs_iso_symm}. We then use the statement of~\cref{lem:meas_hs_iso_symm} to prove Propositions~\ref{prop:HS_Isotropic_general} and~\ref{prop:Measured_HS_Isotro}.
 We define a $d$-dimensional isotropic state as follows: for $p \in [0,1]$
\begin{equation}\label{eq:isotropc_state}
    \zeta^p \coloneqq p~\Phi + (1-p) \Phi^{\perp},
\end{equation}
where
\begin{equation}\label{eq:max_entang_ortho}
    \Phi \coloneqq \frac{1}{d} \sum_{i,j=1}^d |i \rangle\!\langle j| \otimes |i \rangle\!\langle j|  \quad \textnormal{and} \quad \Phi^{\perp} \coloneqq \frac{I_{AB} - \Phi}{d^2-1}.
\end{equation}
Note that $\Phi_{AB}$ and $I_{AB}-\Phi_{AB}$ are orthogonal projections.

\begin{lemma}\label{lem:meas_hs_iso_symm}
    Fix $p,q \in [0,1]$. The following equality holds for all $\gamma \ge 1$:
    \begin{equation}\label{eq:meas_hs_iso_sym}
        E^{\mathcal{M}}_{\gamma}\!\left(\omega^q_{AB}\Vert\omega^p_{AB}\right) = \sup_{M\in \mathcal{M}} \operatorname{Tr}\!\left[M_{AB}\Phi_{AB}\right](q-\gamma p) + \frac{\operatorname{Tr}\!\left[M_{AB}(I_{AB}-\Phi_{AB})\right]}{d^2-1}(1-q-\gamma(1-p)).
    \end{equation}
\end{lemma}
\begin{IEEEproof}
    Consider the following twirling channel:
    \begin{equation}
        \widetilde{\mathcal{T}}_{AB}\!\left(\cdot\right) \coloneqq \int dU (U_A\otimes \overline{U}_B)(\cdot)(U_A\otimes \overline{U}_B)^{\dagger},
    \end{equation}
    where $\overline{U}$ denotes the complex conjugate of $U$. The action of $\widetilde{\mathcal{T}}_{AB}$ on an arbitrary $X_{AB}\in \mathcal{L}(\mathcal{H}_{A}\otimes \mathcal{H}_B)$ results in an operator of the following form~\cite{HH99}:
    \begin{equation}\label{eq:twirled_op_isotropic}
        \widetilde{\mathcal{T}}_{AB}\!\left(X_{AB}\right) = \operatorname{Tr}\!\left[X_{AB}\Phi_{AB}\right]\Phi_{AB} + \frac{\operatorname{Tr}\!\left[X_{AB}(I_{AB}-\Phi_{AB})\right]}{d^2-1}\left(I_{AB} - \Phi_{AB}\right).
    \end{equation}
    It can be easily verified that isotropic states remain invariant under the action of $\widetilde{\mathcal{T}}_{AB}$. Also note that $\widetilde{\mathcal{T}}_{AB}$ is the Hilbert--Schmidt adjoint of itself. The measured hockey-stick divergence between two isotropic states can then be calculated as follows:
    \begin{align}
        E^{\mathcal{M}}_{\gamma}\!\left(\zeta^q_{AB}\Vert\zeta^p_{AB}\right) &= \sup_{M\in \mathcal{M}} \operatorname{Tr}\!\left[M_{AB}\!\left(\zeta^q_{AB} - \gamma \zeta^p_{AB}\right)\right]\\
        &= \sup_{M\in \mathcal{M}} \operatorname{Tr}\!\left[M_{AB}\!\left(\widetilde{\mathcal{T}}_{AB}\!\left(\zeta^q_{AB}\right) - \gamma \widetilde{\mathcal{T}}_{AB}\!\left(\zeta^p_{AB}\right)\right)\right]\\
        &= \sup_{M\in \mathcal{M}} \operatorname{Tr}\!\left[\widetilde{\mathcal{T}}_{AB}\!\left(M_{AB}\right)\!\left(\zeta^q_{AB} - \gamma \zeta^p_{AB}\right)\right]\\
        &= \sup_{M\in \mathcal{M}} \operatorname{Tr}\!\left[M_{AB}\Phi_{AB}\right](q-\gamma p) + \frac{\operatorname{Tr}\!\left[M_{AB}\!\left(I_{AB} - \Phi_{AB}\right)\right]}{d^2-1}(1-q-\gamma(1-p)),
    \end{align}
    where the second equality follows from the invariance of isotropic states under the twirling channel $\widetilde{\mathcal{T}}_{AB}$, the third equality follows from the fact that $\widetilde{\mathcal{T}}_{AB}$ is self-adjoint, and the final equality follows from~\eqref{eq:twirled_op_isotropic} and the definition of isotropic states in~\eqref{eq:isotropc_state}.
\end{IEEEproof}

\begin{proposition}[Hockey-Stick Divergence for Isotropic States]\label{prop:HS_Isotropic_general}
    Let $p,q \in [0,1]$. We have that for $\gamma \geq 1$
\begin{equation}
    E_\gamma(\zeta^q \Vert \zeta^p) 
    =\max\{0,q-\gamma p, (1-q) -\gamma (1-p)\}.
\end{equation} 
\end{proposition}
\begin{IEEEproof}
   For $\gamma \ge 1$, at most one of the two quantities can be positive: $q-\gamma p$ or $1-q-\gamma (1-p)$. The equality in~\eqref{eq:meas_hs_iso_sym} leads to the following:
\begin{align}
    E_{\gamma}\!\left(\zeta^q_{AB}\Vert\zeta^p_{AB}\right) &= \sup_{0 \le M \le I} \operatorname{Tr}\!\left[M_{AB}\Phi_{AB}\right](q-\gamma p) + \frac{\operatorname{Tr}\!\left[M_{AB}\!\left(I_{AB}-\Phi_{AB}\right)\right]}{d^2-1}(1-q-\gamma(1-p))\label{eq:hs_iso_symm}\\
    &\le \sup_{0\le M\le I}\max\left\{\operatorname{Tr}\!\left[M_{AB}\Phi_{AB}\right](q-\gamma p), \frac{\operatorname{Tr}\!\left[M_{AB}\!\left(I_{AB}-\Phi_{AB}\right)\right]}{d^2-1}(1-q-\gamma(1-p))\right\}\\
    &\le \max\left\{0, q-\gamma p, 1-q-\gamma(1-p)\right\},\label{eq:hs_iso_ub}
\end{align}
where the first inequality follows by ignoring the negative term in~\eqref{eq:hs_iso_symm} and the last inequality follows from the fact that $0\le M_{AB}\le I_{AB}$.

To show that the inequality in~\eqref{eq:hs_iso_ub} is saturated, consider $\Phi_{AB}$ to be a specific choice for the measurement operator in~\eqref{eq:hs_iso_symm}, which implies $E_{\gamma}\!\left(\zeta^q_{AB}\Vert\zeta^p_{AB}\right)\ge q-\gamma p$. Furthermore, choosing $M_{AB} = I_{AB} - \Phi_{AB}$ leads to $E_{\gamma}\!\left(\zeta^q_{AB}\Vert\zeta^p_{AB}\right)\ge 1-q-\gamma(1-p)$, and choosing $M_{AB} = 0$ leads to $E_{\gamma}\!\left(\zeta^q_{AB}\Vert\zeta^p_{AB}\right)\ge 0$. Combining the three inequalities, we arrive at the following inequality:
\begin{equation}
    E_{\gamma}\!\left(\zeta^q_{AB}\Vert\zeta^p_{AB}\right)\ge \max\left\{0,q-\gamma p,1-q- \gamma(1-p)\right\},
\end{equation}
which completes the proof.
\end{IEEEproof}

\begin{proposition}[Measured Hockey-Stick Divergence for Isotropic States]
\label{prop:Measured_HS_Isotro}

Let $p,q \in [0,1]$.
     We have the following equality for $\gamma \geq 1$ and $\cM \in \left\{ \cM_{\operatorname{LO}^\star}, \cM_{\operatorname{1W-LOCC}}, \cM_{\operatorname{LOCC}}, \cM_{\operatorname{PPT}}\right\}$:
    \begin{align}
&E_\gamma^{\cM}(\zeta^q \Vert \zeta^p)=  \max\Big\{0, q-\gamma p+ \frac{(1-q)- \gamma (1-p)}{d+1}, 
   \frac{d}{d+1} \left((1-q) -\gamma (1-p) \right) \Big\}.
    \end{align}
\end{proposition}

\begin{IEEEproof}
    \underline{Lower bound:}
 Recall that for $\gamma \geq 1$
\begin{equation}
    E_\gamma^{\operatorname{LO}^\star}(\rho \Vert \sigma) \coloneqq \sup_{M \in \operatorname{LO}^\star} \Tr\!\left[ M (\rho -\gamma \sigma)\right].
\end{equation}
Since the measurement operators $M=0$ and $I-M =I$, trivially belong to the set of local operators, we have
\begin{equation}
    E_\gamma^{\operatorname{LO}^\star}(\zeta^q \Vert \zeta^p) \geq 0 \label{eq:M_o_bound_2}.
\end{equation}

Observe that $M= \sum_{i=1}^d |i \rangle\!\langle i| \otimes |i \rangle\!\langle i|$  and $I-M= \sum_{i\neq j} |i \rangle\!\langle i | \otimes |j\rangle \! \langle j| =I- \sum_{i=1}^d |i \rangle\!\langle i| \otimes |i \rangle\!\langle i|$ belong to the set $\operatorname{LO}^\star$ operators. 
With the former, we arrive at the inequality 
\begin{align}
    E_\gamma^{\operatorname{LO}^\star}(\zeta^q \Vert \zeta^p) &\geq \Tr\!\left[\sum_{i=1}^d |i \rangle\!\langle i| \otimes |i \rangle\!\langle i| \left( (q-\gamma p) \Phi + \left( (1-q)- \gamma (1-p) \right) \Phi^\perp \right) \right] \\
    &= (q-\gamma p)\operatorname{Tr}\!\left[\left(\sum_{i=1}^d |i \rangle\!\langle i|_A \otimes |i \rangle\!\langle i|_B\right)\Phi_{AB}\right] + \left( (1-q)- \gamma (1-p) \right)\operatorname{Tr}\!\left[\left(\sum_{i=1}^d |i \rangle\!\langle i|_A \otimes |i \rangle\!\langle i|_B\right)\Phi^{\perp}_{AB}\right]\\
    &= (q-\gamma p) +\left( (1-q)- \gamma (1-p) \right) \frac{1}{d+1}, \label{eq:one_M_bound_2}
\end{align}
where the first equality followed by linearity of trace operator and by substituting~\eqref{eq:max_entang_ortho}.

Furthermore, with the choice $I- M= \sum_{i\neq j} |i \rangle\!\langle i | \otimes |j\rangle \! \langle j|$, we arrive at another lower bound as follows: 
\begin{align}
    E_\gamma^{\operatorname{LO}^\star}(\zeta^q \Vert \zeta^p) &\geq \Tr\!\left[\left( I- \sum_{i=1}^d |i \rangle\!\langle i| \otimes |i \rangle\!\langle i| \right)\left( (q-\gamma p) \Phi + \left( (1-q)- \gamma (1-p) \right) \Phi^\perp \right) \right] \\ 
    &= 1- \gamma -\left( (q-\gamma p) +\left( (1-q)- \gamma (1-p) \right) \frac{1}{d+1}\right) \\
    &= \frac{d}{d+1} \left( (1-q) -\gamma (1-p) \right)\label{eq:one_M_bound_3},
\end{align}
where the first equality follows from~\eqref{eq:one_M_bound_2}.

Combining~\eqref{eq:M_o_bound_2},\eqref{eq:one_M_bound_2} and \eqref{eq:one_M_bound_3}, we obtain that
\begin{equation}
     E_\gamma^{\operatorname{LO}^\star}(\zeta^q \Vert \zeta^p)  \geq   \max\left\{0, q-\gamma p+ \frac{1}{d+1}(1-q-\gamma(1-p)), \frac{d}{d+1} \left( (1-q) -\gamma (1-p) \right)\right\}. \label{eq:LO_bound_2_1}
\end{equation}

\underline{Upper bound:} Let us first consider the following inequalities that hold for every PPT measurement operator:
\begin{align}
    0 &\le T_B\!\left(M_{AB}\right) \le I_{AB}\\
    \implies 0 &\le \frac{1}{2}\operatorname{Tr}\!\left[T_B\!\left(M_{AB}\right)\!\left(I_{AB}-F_{AB}\right)\right]\le \frac{1}{2}\operatorname{Tr}\!\left[I_{AB}-F_{AB}\right]\\
    \implies 0&\le \frac{1}{2}\operatorname{Tr}\!\left[M_{AB}T_B\!\left(I_{AB}-F_{AB}\right)\right] \le \frac{1}{2}\operatorname{Tr}\!\left[I_{AB}-F_{AB}\right]\\
    \implies 0 &\le \frac{1}{2}\operatorname{Tr}\!\left[M_{AB}\!\left(I_{AB} - d\Phi_{AB}\right)\right] \le \frac{1}{2}(d^2-d)\\
    \implies 0&\le \operatorname{Tr}\!\left[M_{AB}\!\left(I_{AB}-\Phi_{AB} + (1-d)\Phi_{AB}\right)\right] \le d^2-d\\
    \implies 0& \le \operatorname{Tr}\!\left[M_{AB}\!\left(I_{AB}-\Phi_{AB}\right)\right] - (d-1)\operatorname{Tr}\!\left[M_{AB}\Phi_{AB}\right] \le d(d-1)\\
    \implies \frac{1}{d+1} &\le \frac{1}{d^2-1}\times \frac{\operatorname{Tr}\!\left[M_{AB}\!\left(I_{AB}-\Phi_{AB}\right)\right]}{\operatorname{Tr}\!\left[M_{AB}\Phi_{AB}\right]} \le \frac{1}{d+1} + \frac{d}{(d+1)\operatorname{Tr}\!\left[M_{AB}\Phi_{AB}\right]},\label{eq:meas_proj_ratio_iso}
\end{align}
where $F_{AB}$ is the swap operator defined just after~\eqref{eq:werner_extremes}. In the set of inequalities mentioned above, the second line follows from the fact that $\frac{1}{2}(I_{AB}-F_{AB})$ is a projector, the third line follows from the fact that the partial transpose is a self-adjoint map, the fourth line follows from the fact that the partial transpose of the swap operator is the unnormalized maximally entangled vector, and the last line follows by rearranging the terms.

Recall the equality in~\eqref{eq:meas_hs_iso_sym}, which leads to the following:
\begin{align}
    E_{\gamma}^{\operatorname{PPT}}\!\left(\zeta^q_{AB}\Vert\zeta^p_{AB}\right) &= \sup_{\substack{0\le M\le I,\\ 0\le T_B(M)\le I}} \operatorname{Tr}\!\left[M_{AB}\Phi_{AB}\right](q-\gamma p) + \frac{\operatorname{Tr}\!\left[M_{AB}\!\left(I_{AB}-\Phi_{AB}\right)\right]}{d^2-1}(1-q-\gamma (1-p))\\
    &= \sup_{\substack{0\le M\le I,\\ 0\le T_B(M)\le I}} \operatorname{Tr}\!\left[M_{AB}\Phi_{AB}\right]\!\left(q-\gamma p + \frac{\operatorname{Tr}\!\left[M_{AB}\!\left(I_{AB}-\Phi_{AB}\right)\right]}{(d^2-1)\operatorname{Tr}\!\left[M_{AB}\Phi_{AB}\right]}(1-q-\gamma (1-p))\right).\label{eq:PPT_meas_hs_iso_sym_ratio_eq}
\end{align}
Consider the case when $q-\gamma p \ge 0$, which implies that $1-q-\gamma (1-p)\le 0$. The inequality in~\eqref{eq:meas_proj_ratio_iso} yields the following inequality:
\begin{align}
    E_{\gamma}^{\operatorname{PPT}}\!\left(\zeta^q_{AB}\Vert\zeta^p_{AB}\right) &\le \sup_{\substack{0\le M\le I,\\ 0\le T_B(M)\le I}} \operatorname{Tr}\!\left[M_{AB}\Phi_{AB}\right]\!\left(q-\gamma p + \frac{1}{d+1}\!\left(1-q-\gamma (1-p)\right)\right)\\
    &\le \max\left\{q-\gamma p + \frac{1}{d+1}\!\left(1-q-\gamma (1-p)\right),0\right\},\label{eq:PPT_meas_hs_iso_q_ge_gp}
\end{align}
where the last inequality follows from the fact that $M_{AB}\le I_{AB}$ and $M=0$ is a valid choice.

Now consider the case when $1-q-\gamma (1-p) \ge  0$, which implies that $q-\gamma p \le 0$. The equality in~\eqref{eq:PPT_meas_hs_iso_sym_ratio_eq} and the inequality in~\eqref{eq:meas_proj_ratio_iso} yield the following inequality:
\begin{align}
    &E_{\gamma}^{\operatorname{PPT}}\!\left(\zeta^q_{AB}\Vert\zeta^p_{AB}\right)\notag \\ &\le \sup_{\substack{0\le M\le I,\\ 0\le T_B(M)\le I}} \operatorname{Tr}\!\left[M_{AB}\Phi_{AB}\right]\!\left(q-\gamma p + \left(\frac{1}{d+1} + \frac{d}{(d+1)\operatorname{Tr}\!\left[M_{AB}\Phi_{AB}\right]}\right)(1-q-\gamma (1-p))\right)\\
    &= \sup_{\substack{0\le M\le I,\\ 0\le T_B(M)\le I}} \operatorname{Tr}\!\left[M_{AB}\Phi_{AB}\right]\!\left(\left(1-\frac{1}{d+1}\right)(q-\gamma p) - \frac{1}{d+1}(\gamma - 1) + \frac{d}{(d+1)\operatorname{Tr}\!\left[M_{AB}\Phi_{AB}\right]}(1-q-\gamma (1-p))\right)\\
    &\le \sup_{\substack{0\le M\le I,\\ 0\le T_B(M)\le I}} \operatorname{Tr}\!\left[M_{AB}\Phi_{AB}\right]\!\left(\frac{d}{(d+1)\operatorname{Tr}\!\left[M_{AB}\Phi_{AB}\right]}(1-q-\gamma (1-p))\right)\\
    &= \frac{d}{d+1}(1-q-\gamma(1-p)),\label{eq:PPT_meas_hs_iso_q_le_gp}
\end{align}
where we arrived at the last inequality by ignoring the non-positive terms.

Finally, if both $q-\gamma p<0 $ and $1-q-\gamma (1-p)<0$, then the optimal choice is $M_{AB} = 0$, which leads to $E^{\operatorname{PPT}}_{\gamma}\!\left(\zeta^q_{AB}\Vert\zeta^p_{AB}\right) = 0$. Combining this fact with~\eqref{eq:PPT_meas_hs_iso_q_ge_gp} and~\eqref{eq:PPT_meas_hs_iso_q_le_gp} leads to the following inequality:
\begin{equation}
    E^{\operatorname{PPT}}_{\gamma}\!\left(\zeta^q_{AB}\Vert\zeta^p_{AB}\right) \le \max\left\{0, q-\gamma p + \frac{1}{d+1}(1-q-\gamma(1-p)), \frac{d}{d+1}(1-q-\gamma(1-p))\right\}.
\end{equation}
Recall that the right-hand side of the above inequality is a lower bound on $E^{\operatorname{LO}^{\star}}_{\gamma}\!\left(\zeta^q_{AB}\Vert\zeta^p_{AB}\right)$ in~\eqref{eq:LO_bound_2_1}. Since
\begin{equation}
    E^{\operatorname{LO}^{\star}}_{\gamma}\!\left(\zeta^q_{AB}\Vert\zeta^p_{AB}\right) \le E^{\operatorname{1W-LOCC}}_{\gamma}\!\left(\zeta^q_{AB}\Vert\zeta^p_{AB}\right) \le E^{\operatorname{LOCC}}_{\gamma}\!\left(\zeta^q_{AB}\Vert\zeta^p_{AB}\right) \le E^{\operatorname{PPT}}_{\gamma}\!\left(\zeta^q_{AB}\Vert\zeta^p_{AB}\right),
\end{equation}
we conclude that all of the above quantities are equal to each other, and their value is given by $\max\left\{0, q-\gamma p + \frac{1}{d+1}(1-q-\gamma(1-p)), \frac{d}{d+1}(1-q-\gamma(1-p))\right\}$.
\end{IEEEproof}

\begin{remark}[High Dimensions]
    When $d \to \infty$, we see that the quantity in~\cref{prop:Measured_HS_Isotro} is equal to the quantity in~\cref{prop:HS_Isotropic_general}. 
\end{remark}

\subsection{Channel Divergence with All Possible Measurements} \label{App:channel_div_all_meas}
Let $\bar{\mathcal{M}}$ denote the set of all measurements. Then,
\begin{equation}
    E_\gamma(\cP \Vert \cQ) \equiv E_\gamma^{\bar{\cM}}(\cP \Vert \cQ).
\end{equation}
\begin{proposition}[Properties of Hockey-Stick Divergence for Channels]
The hockey-stick divergence for channels satisfies the following properties:
\begin{enumerate}
   \item 
   Quasi Sub-Additivity: Let $\cP_i$ and $\cQ_i$ for $i\in \{0,1\}$ be channels such that $\cP_0$ and $\cQ_0$ are linear mappings from $A$ to $A'$ and $\cP_1$ and $\cQ_1$ are linear mappings from $A'$ to $B$.
   Also, let $\gamma_1,\gamma_2 \geq 1$. Then,
    \begin{align}
        & E_{\gamma_1 \gamma_2}(\cP_1 \circ \cP_0 \Vert \cQ_1 \circ \cQ_0) 
        \leq \min\big\{ E_{\gamma_1}(\cP_0 \Vert \cQ_0) + \gamma_1 E_{\gamma_2}(\cP_1 \Vert \cQ_1),  E_{\gamma_1}(\cP_1 \Vert \cQ_1) + \gamma_1 E_{\gamma_2}(\cP_0 \Vert \cQ_0)\big\}.
    \end{align}
\end{enumerate} 
\end{proposition}
\begin{IEEEproof}
    \underline{Quasi sub-additivity:}
     Let $\cP_i$ and $\cQ_i$ for $i\in \{0,1\}$ be channels such that $\cP_0$ and $\cQ_0$ are linear mappings from $A$ to $A'$ and $\cP_1$ and $\cQ_1$ are linear mappings from $A'$ to $B$. Also, $R$ is a reference system isomorphic to $A$.

    Consider that
    \begin{align}
        &E_{\gamma_1 \gamma_2}(\cP_1 \circ \cP_0 \Vert \cQ_1 \circ \cQ_0) \notag \\
        &=\sup_{\rho_{RA}} E_{\gamma_1 \gamma_2}\!\left(\cP_1 \circ \cP_0 (\rho_{RA}) \Vert \cQ_1 \circ \cQ_0(\rho_{RA})  \right) \\
        &\leq \sup_{\rho_{RA}} E_{\gamma_1}\!\left(\cP_1 \circ \cP_0 (\rho_{RA}) \Vert \cP_1 \circ \cQ_0(\rho_{RA})  \right) + \gamma_1 E_{\gamma_2}\!\left(\cP_1 \circ \cQ_0 (\rho_{RA}) \Vert \cQ_1 \circ \cQ_0(\rho_{RA})  \right)\\
        &\leq \sup_{\rho_{RA}} E_{\gamma_1}\!\left(\cP_0 (\rho_{RA}) \Vert \cQ_0(\rho_{RA})  \right) + \gamma_1 E_{\gamma_2}\!\left(\cP_1 \circ \cQ_0 (\rho_{RA}) \Vert \cQ_1 \circ \cQ_0(\rho_{RA})  \right) \\
        &\leq \sup_{\rho_{RA}} E_{\gamma_1}\!\left(\cP_0 (\rho_{RA}) \Vert \cQ_0(\rho_{RA})  \right) + \gamma_1 \sup_{\rho_{RA}} E_{\gamma_2}\!\left(\cP_1 \circ \cQ_0 (\rho_{RA}) \Vert \cQ_1 \circ \cQ_0(\rho_{RA})  \right) \\
        &= E_{\gamma_1}(\cP_0 \Vert \cQ_0)+ \gamma_1 \sup_{\rho_{RA}} E_{\gamma_2}\!\left(\cP_1 \circ \cQ_0 (\rho_{RA}) \Vert \cQ_1 \circ \cQ_0(\rho_{RA})  \right) \\
        & \le E_{\gamma_1}(\cP_0 \Vert \cQ_0)+ \gamma_1 \sup_{\sigma_{RA'}} E_{\gamma_2}\!\left(\cP_1 (\sigma_{RA'}) \Vert \cQ_1 (\sigma_{RA'})  \right) \\
        &= E_{\gamma_1}(\cP_0 \Vert \cQ_0) + \gamma_1 E_{\gamma_2}(\cP_1 \Vert \cQ_1),
    \end{align}
    where the first equality follows from the triangular property of the hockey-stick divergence~\cite[Eq~II.16]{hirche2023quantum}; second inequality from  the data processing of the hockey-stick divergence~\cite[Lemma~4]{sharma2012strong}. 

The second expression can be obtained by choosing $\cQ_1 \circ \cP_0 (\rho_{RA})$ instead of $\cP_1 \circ \cQ_0 (\rho_{RA})$ in the second equality above and proceeding with the similar decompositions and arguments.
\end{IEEEproof}

\medskip
We show that the hockey-stick channel divergence and the PPT measured hockey-stick channel divergence are SDP computable. 
\begin{proposition}[SDP for $E_\gamma$ Channel Divergence] \label{prop:E_gamma_channel}
    Let $\cP$ and $\cQ$ be two quantum channels, and let $\gamma \geq 1$. Then, the channel divergence $E_\gamma(\cP \Vert \cQ)$ is equivalent to the following expression:
    \begin{align}
        &E_\gamma(\cP \Vert \cQ) =
        \sup_{ \Omega_{RB} \geq 0, \rho_{R} \geq 0} \left\{ 
     \begin{array}
[c]{c}
\Tr\!\left[ \Omega_{RB} (\Gamma_{RB}^\cP - \gamma \Gamma_{RB}^\cQ )\right] :  
\Tr[ \rho_R]=1, \ \Omega_{RB} \leq \rho_R \otimes I_B 
\end{array}  
        \right\},
    \end{align}
    where $\Gamma_{RB}^\cN$ is the Choi operator of the channel $\cN_{A \to B}$.

    Furthermore, its dual expression evaluates to the following:
    \begin{equation}
        E_\gamma(\cP \Vert \cQ)= \inf_{\mu \geq 0, Z_{RB} \geq 0} \left\{ 
        \begin{array}
[c]{c}
        \mu : 
        \Gamma_{RB}^\cP - \gamma \Gamma_{RB}^\cQ \leq Z_{RB}, \  \mu I_{RB} \geq \Tr_B[ Z_{RB}] 
        \end{array}
        \right\}.
    \end{equation}
\end{proposition}
\begin{IEEEproof}
By the joint-convexity of hockey-stick divergence~\cite[Proposition~II.5]{hirche2023quantum} together with the Schmidt decomposition, we have that the supremum in the channel divergence is achieved by pure states.
   Then, together with the variational form of hockey-stick divergence, we have that
    \begin{equation}
E_\gamma\!\left(\cP \Vert \cQ \right) = \sup_{\phi_{RA} \in \cD}  \sup_{0\leq M_{RB} \leq I} \Tr\!\left[ M_{RB} \left(\cP_{A \to B}(\phi_{RA})- \gamma \cQ_{A \to B}(\phi_{RA})\right)\right],
    \end{equation}
where $\phi_{RA}$ is a pure state with the system $R$ isomorphic to system $A$.

We can rewrite the above as follows:
\begin{equation}
E_\gamma\!\left(\cP \Vert \cQ \right) = \sup_{\substack{\phi_{RA} \geq 0, \\ M_{RB} \geq 0}}  \left\{\Tr\!\left[ M_{RB} \left(\cP_{A \to B}(\phi_{RA})- \gamma \cQ_{A \to B}(\phi_{RA})\right)\right] : \begin{aligned}
   & \Tr[\phi_{RA}]= 1, \ \Tr[\phi_{RA}^2]=1, \\
   & M_{RB} \leq I  
\end{aligned}\right\},
    \end{equation}
where $\phi_{RA}$ is a pure bipartitie state that satisfies $\Tr[\phi_{RA}]= 1, \ \Tr[\phi_{RA}^2]=1,  \phi_{RA} \geq 0$.
Note also that is is equivalent to 
\begin{equation}
E_\gamma\!\left(\cP \Vert \cQ \right) = \sup_{\substack{\phi_{RA} \geq 0, \\ M_{RB} \geq 0}}  \left\{\Tr\!\left[ M_{RB} \left(\cP_{A \to B}(\phi_{RA})- \gamma \cQ_{A \to B}(\phi_{RA})\right)\right] : \begin{aligned}
   & \Tr[\phi_{RA}]= 1, \ \Tr[\phi_{RA}^2]=1, \\
   & M_{RB} \leq I, \ \phi_R > 0 
\end{aligned}\right\}, \label{eq:E_ga_rewritten_pure}
    \end{equation}
due to the fact that the set of pure states with reduced state $\phi_R$ positive definite is dense in the set of all pure states. Then, any such pure state can be written as 
\begin{equation} \label{eq:pure_state_with_X_R}
    \phi_{RA} = X_R \Gamma_{RA} X_R^\dag 
\end{equation}
for some linear operator $X$ such that $\Tr [X_R^\dag X_R] = 1$ and $|X_R| > 0$, where $\Gamma_{RA}$ is the unnormalized  maximally entangled operator defined just after~\eqref{eq:gamma_vec_defn}.

Using this, we find that the objective function can be rewritten as
\begin{align}
    &\Tr\!\left[ M_{RB} \left(\cP_{A \to B}(\phi_{RA})- \gamma \cQ_{A \to B}(\phi_{RA})\right)\right] \notag \\ 
    &= \Tr\!\left[ M_{RB} \left(\cP_{A \to B}- \gamma \cQ_{A \to B}\right)(X_R \Gamma_{RA} X_R^\dag)\right] \\
    &=\Tr\!\left[ X_R^\dag M_{RB} X_R \left(\cP_{A \to B}- \gamma \cQ_{A \to B} \right)(\Gamma_{RA} )\right] \\ 
    &=\Tr\!\left[ X_R^\dag M_{RB} X_R \left( \Gamma_{RB}^\cP- \gamma \Gamma_{RB}^\cQ \right)\right]. \label{eq:Objective_fun}
\end{align}

Also, the following equivalence holds:
\begin{equation}
    0 \leq M_{RB} \leq I_{RB} \Longleftrightarrow 0 \leq X_R^\dag M_{RB} X_R \leq X_R^\dag X_R \otimes I_B. 
\end{equation}
With that we choose $\Omega_{RB} \coloneqq X_R^\dag M_{RB} X_R $, $\rho_R \coloneqq X_R^\dag X_R$ since we have $\Tr[X_R^\dag X_R]=1$.

Using the substitutions selected along with~\eqref{eq:E_ga_rewritten_pure} and~\eqref{eq:Objective_fun}, we arrive at 
\begin{equation}
     E_\gamma(\cP \Vert \cQ) =\sup_{ \Omega_{RB} \geq 0, \rho_{R} \geq 0} \left\{ \Tr\!\left[ \Omega_{RB} (\Gamma_{RB}^\cP - \gamma \Gamma_{RB}^\cQ )\right] : \Tr[ \rho_R]=1, \ \Omega_{RB} \leq \rho_R \otimes I_B \right\}.
\end{equation}
This completes the proof of the primal formulation.

To obtain the dual representation, consider 
\begin{align}
  &E_\gamma(\cP \Vert \cQ) \\ \notag
  &=\sup_{ \Omega_{RB} \geq 0, \rho_{R} \geq 0} \left\{ \Tr\!\left[ \Omega_{RB} (\Gamma_{RB}^\cP - \gamma \Gamma_{RB}^\cQ )\right] + \inf_{\mu \in \mathbb{R}, Z_{RB} \geq 0} \mu(1-\Tr[ \rho_R]) + \Tr\!\left[Z_{RB} (\rho_R \otimes I_B-\Omega_{RB}) \right]\right\}  \\
 &=\sup_{ \Omega_{RB} \geq 0, \rho_{R} \geq 0}  \inf_{\mu \in \mathbb{R}, Z_{RB} \geq 0}\left\{ \Tr\!\left[ \Omega_{RB} (\Gamma_{RB}^\cP - \gamma \Gamma_{RB}^\cQ )\right] + \mu(1-\Tr[ \rho_R]) + \Tr\!\left[Z_{RB} (\rho_R \otimes I_B-\Omega_{RB}) \right]\right\} \\
 &\leq  \inf_{\mu \in \mathbb{R}, Z_{RB} \geq 0} \sup_{ \Omega_{RB} \geq 0, \rho_{R} \geq 0} \left\{ \Tr\!\left[ \Omega_{RB} (\Gamma_{RB}^\cP - \gamma \Gamma_{RB}^\cQ )\right] + \mu(1-\Tr[ \rho_R]) + \Tr\!\left[Z_{RB} (\rho_R \otimes I_B-\Omega_{RB}) \right]\right\}\\
 &=\inf_{\mu \in \mathbb{R}, Z_{RB} \geq 0} \sup_{ \Omega_{RB} \geq 0, \rho_{R} \geq 0} \left\{ \mu + \Tr\!\left[ \Omega_{RB} (\Gamma_{RB}^\cP - \gamma \Gamma_{RB}^\cQ - Z_{RB})\right] + \Tr\!\left[\rho_R\left( -\mu I +\Tr_B[Z_{RB}]\right)\right]
 \right\} \\
 &= \inf_{\mu \in \mathbb{R}, Z_{RB} \geq 0} \left\{\mu + \sup_{ \Omega_{RB} \geq 0, \rho_{R} \geq 0}  \Tr\!\left[ \Omega_{RB} (\Gamma_{RB}^\cP - \gamma \Gamma_{RB}^\cQ - Z_{RB})\right] + \Tr\!\left[\rho_R\left( -\mu I +\Tr_B[Z_{RB}]\right)\right] \right\}\\
 &= \inf_{\mu \in \mathbb{R}, Z_{RB} \geq 0} \left\{\mu : \Gamma_{RB}^\cP - \gamma \Gamma_{RB}^\cQ  \leq Z_{RB}, \ \mu I \geq \Tr_B[Z_{RB}]\right\} \\
 &=\inf_{\mu \geq 0, Z_{RB} \geq 0} \left\{\mu : \Gamma_{RB}^\cP - \gamma \Gamma_{RB}^\cQ  \leq Z_{RB}, \ \mu I \geq \Tr_B[Z_{RB}]\right\},
\end{align}
where the last inequality follows since $Z_{RB} \geq 0$ and then $\mu I \geq \Tr_B[Z_{RB}]$ holds only when $\mu \geq 0$.

Now, what is remaining is to show that the strong duality holds. To see this, choose  $Z_{RB} =\Gamma_{RB}^\cP - \gamma \Gamma_{RB}^\cQ +\delta I_{RB}$ and  $\mu$ such that $ \mu I_R = \Tr_B[Z_{RB}] +\delta I_R $ together form a strictly feasible point for all $\delta >0$ and a feasible point for the primal (i.e., supremum formulation) is $\rho_R=\pi_R$, which is the maximally mixed state (in the Hilbert space $\cH_R$) and $\Omega_{RB}= \pi_R \otimes I_B$.
This concludes the proof.
\end{IEEEproof}

\begin{proposition}[Channel Divergence for Depolarizing Channels]\label{prop:all_meas_Depolarized_channel}
    Let $p,q \in [0,1]$ and $\gamma \geq 1$. We have that 
    \begin{align}
       & E_\gamma\!\left( \cA_{\operatorname{Dep}}^q \Vert \cA_{\operatorname{Dep}}^p \right)=
       \max \left\{  0, (1-q) -\gamma (1-p) + \frac{q-\gamma p}{d^2}, q-\gamma p -\frac{q-\gamma p}{d^2} \right\},
    \end{align}
    where $\cA_{\operatorname{Dep}}^p$ is the depolarizing channel with parameter $p$ 
    and $d$ is the dimension of the input space of the channel $\cA_{\operatorname{Dep}}^p$.
\end{proposition}
\begin{IEEEproof}
   Since $E_\gamma$ satisfies data-processing, it is a generalized divergence. Also, with the fact that depolarizing channels are jointly-covariant, using~\cite[Proposition~7.84]{KW20}, we have that 
\begin{align}
      E_\gamma\!\left( \cA_{\operatorname{Dep}}^q \Vert \cA_{\operatorname{Dep}}^p \right) = E_\gamma\!\left( \frac{\Gamma_{RB}^{\cA_{\operatorname{Dep}}^q}}{d} \middle \Vert \frac{\Gamma_{RB}^{\cA_{\operatorname{Dep}}^p}}{d} \right),
\end{align}
where $\Gamma_{RB}^{\cA_{\operatorname{Dep}}^p}$ is the unnormalized Choi-state of the channel $\cA_{\operatorname{Dep}}^p$.
By using the fact that 
\begin{equation}
\Gamma_{RB}^{\cA_{\operatorname{Dep}}^p} = (1-p)\Gamma_{RA} + p \frac{I_{d^2}}{d},
\end{equation}
consider 
\begin{align}
\frac{\Gamma_{RB}^{\cA_{\operatorname{Dep}}^p}}{d} &= (1-p) \frac{\Gamma_{RA}}{d} + \frac{p}{d^2} I_{d^2} \\ 
&= (1-p) \Phi + \frac{p}{d^2} \left( I_{d^2} - \Phi + \Phi \right) \\
&= \left( 1-p + \frac{p}{d^2}\right) \Phi + \frac{p}{d^2} (d^2-1) \frac{(I_{d^2} - \Phi)}{d^2-1} \\
&= \zeta^\eta \label{eq:iso_equi_depo}
\end{align}
where $ \eta \coloneqq 1-p + \frac{p}{d^2}$ in~\eqref{eq:isotropc_state}.

With that, by choosing $\eta' \coloneqq  1-q + \frac{q}{d^2}$, we arrive at 
\begin{equation}
  E_\gamma\!\left( \cA_{\operatorname{Dep}}^q \Vert \cA_{\operatorname{Dep}}^p \right)= E_\gamma\!\left( \zeta^{\eta'} \middle \Vert \zeta^\eta \right).  
\end{equation}
Then, we conclude the proof by applying~\cref{prop:HS_Isotropic_general}.
\end{IEEEproof}

\subsection{Proof of~\cref{prop:channel_div_PPT}} \label{App:channel_div_PPT_proof}

    The proof of the supremum formulation (primal) follows similar to the proof of~\cref{prop:E_gamma_channel} with the added constraints as given below. 
    For PPT measurements, we require 
    \begin{equation}
        0 \leq T_B(M_{RB}) \leq I.
    \end{equation}
This is equivalent to enforcing 
\begin{equation}
    0 \leq X_R^\dag T_B(M_{RB}) X_R \leq X_R^\dag X_R \otimes I_B \quad \Longleftrightarrow \quad 0 \leq T_B(\Omega_{RB}) \leq \rho_R \otimes I_B,
\end{equation}
with the choices $\Omega_{RB} \coloneqq X_R^\dag M_{RB} X_R $ and $\rho_R \coloneqq X_R^\dag X_R$.

For the derivation of the dual, consider that
\begin{align}
  &E_\gamma^{\operatorname{PPT}}(\cP \Vert \cQ) \\ \notag
  &=\sup_{ \Omega_{RB} \geq 0, \rho_{R} \geq 0} \left\{ \Tr\!\left[ \Omega_{RB} (\Gamma_{RB}^\cP - \gamma \Gamma_{RB}^\cQ )\right] + \inf_{\substack{\mu \in \mathbb{R}, Z_{RB} \geq 0 \\ L_{RB}, Y_{RB} \geq 0}} \left\{ \begin{aligned}
      &\mu(1-\Tr[ \rho_R]) + \Tr\!\left[Z_{RB} (\rho_R \otimes I_B-\Omega_{RB}) \right]  \\
  & +\Tr\!\left[ L_{RB} T_B(\Omega_{RB}) \right] \\&+\Tr\!\left[Y_{RB} \left(\rho_R \otimes I_B-T_B(\Omega_{RB})\right) \right] 
  \end{aligned} \right\}\right\}  \\
 &=\sup_{ \Omega_{RB} \geq 0, \rho_{R} \geq 0} \inf_{\substack{\mu \in \mathbb{R}, Z_{RB} \geq 0 \\ L_{RB}, Y_{RB} \geq 0}} \left\{ \begin{aligned} &\Tr\!\left[ \Omega_{RB} (\Gamma_{RB}^\cP - \gamma \Gamma_{RB}^\cQ )\right] + \mu(1-\Tr[ \rho_R]) + \Tr\!\left[Z_{RB} (\rho_R \otimes I_B-\Omega_{RB}) \right] \\ & +\Tr\!\left[ L_{RB} T_B(\Omega_{RB}) \right] +\Tr\!\left[Y_{RB} \left(\rho_R \otimes I_B-T_B(\Omega_{RB})\right) \right] 
 \end{aligned}
 \right\} \\
 &\leq \inf_{\substack{\mu \in \mathbb{R}, Z_{RB} \geq 0 \\ L_{RB}, Y_{RB} \geq 0}}\sup_{ \Omega_{RB} \geq 0, \rho_{R} \geq 0}  \left\{ \begin{aligned} &\Tr\!\left[ \Omega_{RB} (\Gamma_{RB}^\cP - \gamma \Gamma_{RB}^\cQ )\right] + \mu(1-\Tr[ \rho_R]) + \Tr\!\left[Z_{RB} (\rho_R \otimes I_B-\Omega_{RB}) \right] \\ & +\Tr\!\left[ L_{RB} T_B(\Omega_{RB}) \right] +\Tr\!\left[Y_{RB} \left(\rho_R \otimes I_B-T_B(\Omega_{RB})\right) \right] 
 \end{aligned}
 \right\}\\
 &= \inf_{\substack{\mu \in \mathbb{R}, Z_{RB} \geq 0 \\ L_{RB}, Y_{RB} \geq 0}}\left\{\mu + \sup_{ \Omega_{RB} \geq 0, \rho_{R} \geq 0} \begin{aligned}
  &\Tr\!\left[ \Omega_{RB} (\Gamma_{RB}^\cP - \gamma \Gamma_{RB}^\cQ - Z_{RB} +T_B(L_{RB}) -T_B(Y_{RB}))\right] \\  &+ \Tr\!\left[\rho_R\left( -\mu I +\Tr_B[Z_{RB}] +\Tr_B[Y_{RB}]\right)\right]  
 \end{aligned} \right\}\\
 &= \inf_{\substack{\mu \in \mathbb{R}, Z_{RB} \geq 0 \\ L_{RB}, Y_{RB} \geq 0}} \left\{\mu : \Gamma_{RB}^\cP - \gamma \Gamma_{RB}^\cQ  \leq Z_{RB} +T_B(Y_{RB})-T_B(L_{RB}), \ \mu I \geq \Tr_B[Z_{RB}]+ \Tr_B[Y_{RB}]\right\} \\
 &\inf_{\substack{\mu \geq 0, Z_{RB} \geq 0 \\ L_{RB}, Y_{RB} \geq 0}} \left\{\mu : \Gamma_{RB}^\cP - \gamma \Gamma_{RB}^\cQ  \leq Z_{RB} +T_B(Y_{RB})-T_B(L_{RB}), \ \mu I \geq \Tr_B[Z_{RB}]+ \Tr_B[Y_{RB}]\right\},
\end{align}
where the last inequality follows since $Z_{RB}, Y_{RB} \geq 0$ and then $\mu I \geq \Tr_B[Z_{RB}] +\Tr_B[Y_{RB}]$ holds only when $\mu \geq 0$.

Now, what is remaining is to show that the strong duality holds. To see this, choose  $Z_{RB} =\Gamma_{RB}^\cP - \gamma \Gamma_{RB}^\cQ $, $Y_{RB}=\delta_1 I_{RB}$, $L_{RB}= \delta_2 I_{RB}$ and  $\mu$ such that $ \mu I_R = \Tr_B[Z_{RB}] + \Tr_B[Y_{RB}] + \delta_3 I_R $ together form a strictly feasible point for all $\delta_i >0$ with $i \in \{1,2,3\}$ and a feasible point for the primal (i.e., supremum formulation) is $\rho_R=\pi_R$, which is the maximally mixed state (in the Hilbert space $\cH_R$) and $\Omega_{RB}= \pi_R \otimes I_B$.
This concludes the proof.

\subsection{Data Processing under PPT Measurements}

Note that $E_\gamma^{\cM}$ for a measurement class $\cM$ may not satisfy data processing under every quantum channel in general. Next, we show that $E_\gamma^{\operatorname{PPT}}$ satisfies data-processing under special classes that are PPT preserving. 
PPT preserving channels are defined as follows:
A bipartite channel $\cN_{AB \to A'B'}$ is called as PPT preserving channel if $T_{A'} \circ \cN_{AB \to A'B'} \circ T_A$ is a completely positive map.

Note that we highlight what is the bipartition considered when PPT-measured hockey-stick divergences are written, and we omit those systems when it is clear from the context.

\begin{proposition}[Data processing of $E_\gamma^{\operatorname{PPT}}$] \label{prop:data_process_E_gamma_PPT}
    Let $\rho_{AR}$ and $\sigma_{AR}$ be quantum states, and let $\cP_{AR \to B R'}$ be a PPT preserving channel with $A,B$ systems belonging to one party and $R,R'$ systems belonging to the other party. 
    Then 
    \begin{equation}
    \label{eq:ineq_DPI}
E_\gamma^{\operatorname{PPT}(B:R')}\!\left(\cP_{AR \to B R'}(\rho_{AR}) \Vert \cP_{AR \to B R'}(\sigma_{AR}) \right) \leq E_\gamma^{\operatorname{PPT}(A:R)}\!\left( \rho_{AR} \Vert \sigma_{AR}\right).
    \end{equation}
Furthermore, for a local isometric channel $\cW_{R \to R' A'}$ acting on the input system $R$, and a channel $\cN_{A \to B}$ on input system $A$, we also have that 
\begin{align} \label{eq:isometry_PPT_local}
    & E_\gamma^{\operatorname{PPT}(B:R)}\!\left(\cN_{A \to B }(\rho_{AR}) \Vert \cN_{A\to B }(\sigma_{AR}) \right) \notag \\&=E_\gamma^{\operatorname{PPT}(B:R'A')}\!\left( \cW_{R \to R'A'} \otimes \cN_{A \to B }(\rho_{AR}) \Vert \cW_{R \to R'A'} \otimes\cN_{A \to B } (\sigma_{AR}) \right).
\end{align}
\end{proposition}
\begin{IEEEproof}
     Let $M_{BR'}$ be such that $0 \leq M_{BR'} \leq I_{BR'}, 0\leq T_B(M_{BR'}) \leq I_{BR'}$ and consider that 
    \begin{align}
        \Tr\!\left[ M_{BR'} \left( \cP_{AR \to B R'}(\rho_{AR}) - \gamma \cP_{AR \to B R'}(\sigma_{AR})\right)\right] = \Tr\!\left[ \cP^\dag_{BR' \to AR}( M_{BR'}) ( \rho_{AR} -\gamma \sigma_{AR})\right].
    \end{align}
Since $0 \leq M_{BR'}, T_B(M_{BR'}) \leq I_{BR'}$ and $\cP$ is a CPTP map ($\cP^\dag$ is also CPTP and unital), we also have that 
\begin{equation}\label{eq:support_PPT}
    0 \leq \cP^\dag_{BR' \to AR}( M_{BR'}) \leq \cP^\dag_{BR' \to AR}(I_{BR'}) =I_{AR}.
\end{equation}

Also consider 
\begin{align}
 T_A\left( \cP^\dag_{BR' \to AR}( M_{BR'}) \right) &=
 T_A\left(   \cP^\dag_{BR' \to AR} T_B \left( T_B( M_{BR'}) \right) \right) \\
 &= T_A \circ \cP^\dag_{BR' \to AR} \circ T_B \left( T_B( M_{BR'}) \right), 
\end{align}
where the first equality holds since the inverse of $T_B$ is $T_B$ itself.

Let $X_{BR'}$ and $ Y_{AR}$ be arbitrary positive semi-definite operators. Since $T_B\circ\mathcal{P}_{AR \to BR'}\circ T_{A}$ is a positive map, 
\begin{align}
  0 &\leq  \Tr\!\left[ X_{BR'} T_B \left(\cP_{AR \to BR'}\left(T_A(Y_{AR})\right)\right)\right]\\
    &=\Tr\! \left[T_B(X_{BR'}) \cP_{AR \to BR'}\left(T_A(Y_{AR})\right)\right] \\
    &=\Tr\! \left[T_A \circ \cP^\dag_{BR' \to AR} \circ T_B(X_{BR'}) Y_{AR}\right],  
\end{align}
which leads to $T_A \circ \cP^\dag_{BR' \to AR} \circ T_B $ being a positive map. 
With that, we arrive at
\begin{equation}
    0 \leq T_A \circ \cP^\dag_{BR' \to AR} \circ T_B \left( T_B( M_{BR'}) \right) \leq T_A \circ \cP^\dag_{BR' \to AR} \circ T_B (I_{BR'}) =I_{AR}
\end{equation}
with the use of $0 \leq T_B(M_{BR'}) \leq I_{BR'}$ and~\eqref{eq:support_PPT}.
Then, this shows that 
\begin{equation}
 \cP^\dag_{BR' \to AR}( M_{BR'}) \in \left\{M_{AR}:  0 \leq M_{AR} \leq I_{AR}, 0\leq T_A(M_{AR}) \leq I_{AR} \right\}. 
\end{equation}

With that, we have 
\begin{align} 
        \Tr\!\left[ M_{BR'} \left( \cP_{AR \to B R'}(\rho_{AR}) - \gamma \cP_{AR \to B R'}(\sigma_{AR})\right)\right] &= \Tr\!\left[ \cP^\dag_{BR' \to AR}( M_{BR'}) ( \rho_{AR} -\gamma \sigma_{AR})\right] \\
        & \leq \sup_{0 \leq M_{AR}, T_A(M_{AR}) \leq I_{AR}} \Tr\!\left[M_{AR}(\rho_{AR} -\gamma \sigma_{AR}) \right] \\
        &= E_\gamma^{\operatorname{PPT}(A:R)}(\rho_{AR} \Vert \sigma_{AR}).
    \end{align}
Finally, by supremizing over $M_{BR'}$ such that $0 \leq M_{BR'}, T_B(M_{BR'}) \leq I_{BR}$, we conclude the proof of the inequality presented in~\eqref{eq:ineq_DPI}.

The equality in~\eqref{eq:isometry_PPT_local} follows by applying the data-processing inequality under isometric channel $\cW_{R \to R'A'}\otimes\mathcal{I}_B$ and its inverse channel $\cW^{-1}_{R \to R'A'}\otimes\mathcal{I}_B$, both of which belong to the set of PPT preserving channels with respect to the partial transpose on system $B$.
\end{IEEEproof}

\subsection{Proof of~\cref{prop:PPT_measured_joint_covariance}} \label{App:PPT_measured_joint_cov_Proof}
   Let $G$ be a finite group with $\{ U_A(g)\}_{g\in G}$ and $\{V_B(g)\}_{g \in G}$ unitary representations. Since the set $\{ \psi_{RA}: \psi_A= \cT_G(\psi_A)\}$  is a subset of all pure states, we immediately obtain the following inequality:
   \begin{equation}\label{eq:forward_PPT}
        E_\gamma^{\operatorname{PPT}}(\cP \Vert \cQ) \geq \sup_{\psi_{RA}}\left\{ E_\gamma^{\operatorname{PPT}}\!\left(\cP_{A \to B}(\psi_{RA}) \Vert \cQ_{A \to B} (\psi_{RA})\right) :  \psi_A = \cT_G (\psi_A)\right\}.
   \end{equation}
What remains to show is the reverse inequality.

Let $\rho_A \in \cL(\cH_A)$ and $\psi_{RA}$ 
be a purification of state $\rho_A$.  Let $\bar{\rho}_A$ be the group average of $\rho_A$, i.e.; 
\begin{equation}
    \bar{\rho}_A \coloneqq \frac{1}{|G|} \sum_{g \in G} U_A(g) \rho_A U_A^\dag(g),
\end{equation}
and let $\phi_{RA}^{\bar{\rho}}$ be a purification of $\bar{\rho}_A$.

Let us also consider the following state $\psi^{\bar{\rho}}_{PR'A}$, which is also a purification of the state $\bar{\rho}_A$ with the purifying systems $P$ and $R'$ with $P$ system being a classical system and $R'$ being isomorphic to system $R$:
\begin{equation}
    |\psi^{\bar{\rho}} \rangle_{PR'A} \coloneqq \frac{1}{\sqrt{|G|}} \sum_{g \in G} |g\rangle_P \otimes (I_{R'} \otimes U_A(g)) | \psi\rangle_{RA},
\end{equation}
where $\{ |g\rangle \}_{g \in G}$ is an orthonormal basis for $\cH_P$ indexed by the elements of $G$.

Since all purifications of a state can be mapped to each other by isometries acting on the purifying systems, there exists an isometry $W_{R \to PR'}$ that forms the isometric channel $\cW_{R \to PR'}$ such that 
\begin{equation}
   \psi^{\bar{\rho}}_{PRA} = \cW_{R \to PR'}( \phi_{RA}^{\bar{\rho}}).
\end{equation}

Consider that, 
\begin{align}
   & E_\gamma^{\operatorname{PPT}(R:B)}\!\left(\cP_{A \to B} (\phi_{RA}^{\bar{\rho}}) \middle \Vert \cQ_{A \to B} (\phi_{RA}^{\bar{\rho}})\right) \notag \\
   &=  E_\gamma^{\operatorname{PPT}(PR':B)}\!\left( \cW_{R \to PR'} \circ \cP_{A \to B} \circ  (\phi_{RA}^{\bar{\rho}}) \middle \Vert \cW_{R \to PR'} \circ \cQ_{A \to B} (\phi_{RA}^{\bar{\rho}})\right) \\
   &= E_\gamma^{\operatorname{PPT}(PR':B)}\!\left(\cP_{A \to B} (\psi_{PR'A}^{\bar{\rho}}) \middle \Vert \cQ_{A \to B} (\psi_{PR'A}^{\bar{\rho}})\right)  \\
    & \geq E_\gamma^{\operatorname{PPT}(PR':B)}\!\Bigg( \frac{1}{ |G|} \sum_{g \in G} |g\rangle\! \langle g|_P \otimes \left(\cP_{A \to B} \circ \cU_A(g) \right)(\psi_{RA})  
   \Bigg \Vert \frac{1}{ |G|} \sum_{g \in G} |g\rangle\! \langle g|_P \otimes \left(\cQ_{A \to B} \circ \cU_A(g) \right)(\psi_{RA})\Bigg) \\ 
   & \geq E_\gamma^{\operatorname{PPT}(PR':B)}\!\Bigg( \frac{1}{ |G|} \sum_{g \in G} |g\rangle\! \langle g|_P \otimes \left( \left(\cV_B(g)\right)^\dag \circ \cP_{A \to B} \circ \cU_A(g) \right)(\psi_{RA}) 
   \notag  \\ &\quad  \quad \quad \quad \quad \quad  \quad 
   \Bigg \Vert \frac{1}{ |G|} \sum_{g \in G} |g\rangle\! \langle g|_P \otimes \left(\left(\cV_B(g)\right)^\dag \circ \cQ_{A \to B} \circ \cU_A(g) \right)(\psi_{RA})\Bigg)\\
   &=E_\gamma^{\operatorname{PPT}(PR':B)}\!\left( \frac{1}{ |G|} \sum_{g \in G} |g\rangle\! \langle g|_P \otimes  \cP_{A \to B}(\psi_{RA}) \middle \Vert \frac{1}{ |G|} \sum_{g \in G} |g\rangle\! \langle g|_P \otimes \cQ_{A \to B}(\psi_{RA})\right) \\
   &\geq  E_\gamma^{\operatorname{PPT}(R':B)}\!\left(   \cP_{A \to B}(\psi_{RA}) \middle \Vert  \cQ_{A \to B}(\psi_{RA})\right), \label{eq:last_DP_PPT}
\end{align}
where the first two equalities follow from~\eqref{eq:isometry_PPT_local}, the first inequality by applying data processing of $E_\gamma^{\operatorname{PPT}}$ in~\cref{prop:data_process_E_gamma_PPT} for the dephasing channel $X \to \sum_{g \in G} |g \rangle\!\langle g| X|g \rangle\!\langle g|$ on system $P$ (a PPT preserving channel); the second inequality by applying the data processing under the unitary channel given by the unitary $\sum_{g \in G} |g\rangle\!\langle g|_P \otimes V_B^\dag(g)$ \footnote{This is because that one could implement this operation on these states as a classically controlled LOCC operation where the von-Neumann measurement ${|g\rangle\!\langle g|}_{g \in G}$ is applied on the register $P$ followed by a rotation given by the unitary channel $(\cV_B(g))^\dag$. Recall that we chose the system $P$ to be classical, so the mentioned procedure can be followed.}; third equality by the joint-covariance of the channels $\cP$ and $\cQ$ such that 
\begin{equation}
    \left(\cV_B(g)\right)^\dag \circ \cP_{A \to B} \circ \cU_A(g) = \cP_{A \to B}, \quad \left(\cV_B(g)\right)^\dag \circ \cQ_{A \to B} \circ \cU_A(g) = \cQ_{A \to B};
\end{equation}
and finally the last inequality by the data processing under the partial trace channel on the system $P$ by using~\cref{prop:data_process_E_gamma_PPT}.

By definition, the pure state $\phi_{RA}^{\bar{\rho}}$ is such that its reduced state on $A$ is invariant under the channel $\cT_G$. Then, by optimizing over all such pure states, and noting that $R'$ is isomorphic to system $R$, we obtain that 
\begin{align}
   E_\gamma^{\operatorname{PPT}(R:B)}\!\left(   \cP_{A \to B}(\psi_{RA}) \middle \Vert  \cQ_{A \to B}(\psi_{RA})\right)  
   &\leq \sup_{\phi_{RA}}\left\{ E_\gamma^{\operatorname{PPT}}\!\left(\cP_{A \to B}(\phi_{RA}) \Vert \cQ_{A \to B} (\phi_{RA})\right) :  \phi_A = \cT_G (\phi_A)\right\}.
\end{align}
Since the above inequality holds for all pure states $\psi_{RA}$, we obtain the reverse inequality, 
\begin{equation}
     E_\gamma^{\operatorname{PPT}}(\cP \Vert \cQ) \leq \sup_{\psi_{RA}}\left\{ E_\gamma^{\operatorname{PPT}}\!\left(\cP_{A \to B}(\psi_{RA}) \Vert \cQ_{A \to B} (\psi_{RA})\right) :  \psi_A = \cT_G (\psi_A)\right\}. \label{eq:reverse_PPT}
\end{equation}
Then, we conclude the first equality by combining~\eqref{eq:forward_PPT} and~\eqref{eq:reverse_PPT}.

 To prove the next claim, if the representation $\{U_A(g)\}_{g \in G}$ acting on the input space of the channel $\cP$ and $\cQ$ is irreducible, then for every state $\psi_{RA}$ such that $\rho_A= \psi_A$, it holds that $\bar{\rho}_A =I_A /D_A$. Then, since the maximally entangled state is a purification of the maximally mixed state, we let $\phi_{RA}^{\bar{\rho}}= \Phi_{RA}$, which implies that 
\begin{equation}
    E_\gamma^{\operatorname{PPT}(R:B)}\!\left(\cP_{A \to B} (\phi_{RA}^{\bar{\rho}}) \middle \Vert \cQ_{A \to B} (\phi_{RA}^{\bar{\rho}})\right) =E_\gamma^{\operatorname{PPT}(R:B)}\!\left( \cP_{A \to B}( \Phi_{RA}) \Vert \cQ_{A \to B}( \Phi_{RA}) \right).
\end{equation}

Using~\eqref{eq:last_DP_PPT}, we obtain that 
\begin{equation}
    E_\gamma^{\operatorname{PPT}(R:B)}\!\left( \cP_{A \to B}( \Phi_{RA}) \Vert \cQ_{A \to B}( \Phi_{RA}) \right) \geq E_\gamma^{\operatorname{PPT}(R:B)}\!\left(   \cP_{A \to B}(\psi_{RA}) \middle \Vert  \cQ_{A \to B}(\psi_{RA})\right) 
\end{equation}
for all states $\psi_{RA}$. So, by optimizing over all states $\psi_{RA}$ we have that 
\begin{equation}
   E_\gamma^{\operatorname{PPT}}(\cP \Vert \cQ) \leq E_\gamma^{\operatorname{PPT}(R:B)}\!\left( \cP_{A \to B}( \Phi_{RA}) \Vert \cQ_{A \to B}( \Phi_{RA}) \right).
\end{equation}
By choosing the pure state to be the maximally entangled state in~\eqref{eq:def_channel_HS}, we obtain the reverse inequality, concluding the proof of~\cref{prop:PPT_measured_joint_covariance}.

\subsection{Proof of~\cref{prop:channel_div_PPT_depol}} \label{App:channel_div_PPT_depol_proof}
    Proof follows similar to the proof of~\cref{prop:all_meas_Depolarized_channel}. With the use of~\cref{prop:PPT_measured_joint_covariance}, we have that 
    \begin{align}
         E_\gamma^{\operatorname{PPT}}\!\left( \cA_{\operatorname{Dep}}^q \Vert \cA_{\operatorname{Dep}}^p \right) & = E_\gamma^{\operatorname{PPT}}\!\left( \frac{\Gamma_{RB}^{\cA_{\operatorname{Dep}}^q}}{d} \middle \Vert \frac{\Gamma_{RB}^{\cA_{\operatorname{Dep}}^p}}{d} \right) \\
         &= E_\gamma^{\operatorname{PPT}}\!\left( \zeta^{\eta'} \middle \Vert \zeta^\eta \right), \label{eq:last_PPT_depol}
    \end{align}
    where the last equality follows by the substitution in~\eqref{eq:iso_equi_depo} together with $\eta\coloneqq 1-p+p/d^2$ and $\eta'\coloneqq 1-q+q/d^2$.
    
    Finally, we conclude the proof by adapting~\cref{prop:Measured_HS_Isotro} to~\eqref{eq:last_PPT_depol}.
\end{document}